\documentclass[journal, doublecolumn]{IEEEtran}
\usepackage{amsmath,amsfonts}
\usepackage{algorithm}
\usepackage{array}
\usepackage[caption=false,font=normalsize,labelfont=sf,textfont=sf]{subfig}
\usepackage{textcomp}
\usepackage{stfloats}
\usepackage{url}
\usepackage{verbatim}
\usepackage{graphicx}
\usepackage{cite}
\usepackage{bm} 
\usepackage{units}
\usepackage{algpseudocode}
\makeatletter  
\let\NAT@parse\undefined
\makeatletter
\usepackage[breaklinks, colorlinks, 
citecolor=blue,
linkcolor=blue,urlcolor=black]{hyperref}

\newtheorem {theorem}{Theorem}
\newtheorem {proof}{Proof}
\newtheorem {lemma}{Lemma}
\begin{document}

\title{Multi-Static ISAC based on Network-Assisted Full-Duplex Cell-Free Networks: Performance Analysis and Duplex Mode
Optimization}

\author{Fan Zeng, Ruoyun Liu, Xiaoyu Sun, Jingxuan Yu,~\IEEEmembership{Student,~IEEE}, Jiamin Li, Pengchen Zhu,
	Dongming Wang,~\IEEEmembership{Member,~IEEE},
	and Xiaohu You,~\IEEEmembership{Fellow,~IEEE}}



\maketitle

\let\thefootnote\relax\footnotetext{This work was supported in part by the National Key R\&D Program of China under Grant 2021YFB2900300, by the National Natural Science Foundation of China (NSFC) under Grants 61971127, by the Fundamental Research Funds for the Central Universities under Grant 2242022k60006, and by the Major Key Project of PCL (PCL2021A01-2).(\textit{Corresponding author: Jiamin Li.})}
\let\thefootnote\relax\footnotetext{The authors are with National Mobile Communications Research Laboratory, Southeast University, Nanjing, 210096, China (email: \{zengfan, ryliu, 230238231, jingxuanyu, jiaminli, p.zhu, wangdm, xhyu\}@seu.edu.cn). J. Li, D. Wang and X. You are also with Purple Mountain Laboratories, Nanjing 211111, China.}

\begin{abstract}
  Multi-static integrated sensing and communication (ISAC) technology, which can achieve a wider coverage range and avoid self-interference, is an important trend for the future development of ISAC. Existing multi-static ISAC designs are unable to support the asymmetric uplink (UL)/downlink (DL) communication requirements in the scenario while simultaneously achieving optimal sensing performance. This paper proposes a design for multi-static ISAC based on network-assisted full-duplex (NAFD) cell-free networks can well solve the above problems. Under this design, closed-form expressions for the individual comunication rate and localization error rate are derived under imperfect channel state information, which are respectively utilized to assess the communication and sensing performances. Then, we propose a deep Q-network-based accesss point (AP) duplex mode optimization algorithm to obtain the trade-off between communication and sensing from the UL and DL perspectives of the APs. Simulation results demonstrate that the NAFD-based ISAC system proposed in this paper can achieve significantly better communication performance than other ISAC systems while ensuring minimal impact on sensing performance. Then, we validate the accuracy of the derived closed-form expressions. Furthermore, the proposed optimization algorithm achieves performance comparable to that of the exhaustion method with low complexity.
\end{abstract}
\begin{IEEEkeywords}
  Multi-static integrated sensing and communication, network-assisted full-duplex system, access point duplex mode optimization, multi-objective optimization. 
\end{IEEEkeywords}

\section{Introduction}
\IEEEPARstart{I}{ntegrated} sensing and communication (ISAC) has received increasing attention as one of the six typical application scenarios of 6G\cite{Kaushik2024Toward,banafaa20236g}. In recent years, most of the research has stayed on the single-cell ISAC, such as integrated waveform design \cite{Liu2018Waveform}, joint transmission beamforming \cite{Sun2022Optimal}, and joint signal reception \cite{Guo2023Joint}, etc.

	Multi-static ISAC has more advantages than single-cell ISAC. In multi-static ISAC system, some access points (APs) send sensing signals while others receive echoes, which can avoid the problem of self-interference and achieve a wider coverage range. Multi-static sensing can also obtain distance and doppler-related information from different angles of the target, which can obtain more accurate information and effectively improve the accuracy and reliability of sensing. In addition, the multi-static ISAC system can also perform interference coordination. It can schedule some nodes as receiving/transmitting nodes according to the actual needs and conditions of the network, which can better adapt to different environments and task requirements. Therefore, the research in this paper is mainly based on multi-static ISAC.
	
	Rahman et al.\cite{ Rahman2020Framework } developed a new perceptual mobile network framework, which includes the architecture of multi-static ISAC. Huang et al.\cite {Huang2022Coordinated} proposed a coordinated power control design for multi-static ISAC systems aimed at maximizing the signal to interference plus noise ratio (SINR) at the user equipments (UEs) and the target position estimation Cramer-Rao Lower Bound (CRLB). However, the previous studies are based on cellular network implementation of multi-static ISAC. In multi-static ISAC system, the cross-link interference (CLI) within the system is exacerbated due to the increased number of sensing links. In traditional cellular-based ISAC, inter-cell link interference is complicated. In contrast, in the cell-free (CF) systems, signal transmission and reception have a larger DoF, which enables more flexible communication and sensing cooperation and various interference avoidance\cite{Fang2021Cell,Ammar2022User}. Therefore, there have been an increasing number of recent studies to integrate multi-static ISAC into CF networks. 
	
	Sakhnini et al.\cite{Sakhnini2022Target} were the first to propose a communication and sensing architecture combining CF networks and ISAC. In this architecture, a group of APs transmit downlink (DL) sensing signals and the remaining APs receive echo signals for sensing. All signal processing is performed in the central processing unit (CPU). However, the APs in this architecture only support either sensing or communication. Behdad et al. \cite{Behdad2022PowerAF} and Demirhan et al.\cite{demirhan2024cellfree} proposed power allocation strategies and beamforming strategies for CF-ISAC systems, respectively. These systems allow APs to communicate and sense simultaneously. However, these studies were conducted based on perfect channel state information (CSI) and did not consider the case of imperfect CSI. Mao et al. \cite{mao2023communicationsensing} derived a closed-form expression for the spectral efficiency (SE) of the CF-ISAC system, taking into account the uncertainty of the target location and the imperfect CSI information. However, in the previous study, only the case where uplink (UL) or DL communication coexists with sensing is considered, which constrains the system to serving only a subset of UEs at a time. For instance, in \cite{Sakhnini2022Target}, only UL UEs are considered to share resources with sensing. Conversely, in \cite{Behdad2022PowerAF}, \cite{demirhan2024cellfree}, and \cite{ mao2023communicationsensing }, only DL communications are considered. However, in real systems, the communication demands of UEs often present an asymmetry between UL and DL. The studies of Wei et al. \cite{Wei2020Load} and Soret et al. \cite{Soret2019Queueing} demonstrate that supporting only UL or DL on a piece of resource cannot respond quickly to such demands. In fact, in a multi-static ISAC system, some APs send DL sensing signals and others receive UL echoes and the system already has the capability of simultaneous UL and DL communication, i.e., it can send and receive signals at the same time. Nevertheless, to the best of our knowledge, there has been no research on multi-static ISAC systems for coexisting UL and DL communication.
	
	However, when UL and DL communication and sensing are performed simultaneously in the system, it will exacerbate the CLI in the system, which will seriously affect the performance of the system. Therefore, it is essential to suppress the CLI. Sit et al.  \cite{Sit2011Extension} proposed to reconstruct the interfering signal and perform the suppression of the CLI between sensing and communication using the pilot signal to enhance the dynamic range of the radar. Wang et al.\cite{Wang2019Performance} proposed the network-assisted full-duplex (NAFD) technology for the CLI suppression in the systems that only consider UL and DL communication. The NAFD is based on the CF network in which each AP can perform UL reception or DL transmission with joint signal processing at the CPU side. Once the channels between the DL APs and the UL APs have been estimated using the DL pilot signals, the CPU can obtain the DL precoded signals of all UEs in advance. Consequently, the CLI suppression can be achieved in the digital domain. In \cite{Wang2019Performance}, the case of perfect CSI is considered. Li et al.\cite{ Li2021Network} derived a closed-form expression for the SE of NAFD systems under imperfect CSI. However, similar interference cancellation has not been investigated for ISAC systems in which UL and DL communication coexist.
	
	Furthermore, the majority of existing multi-static ISAC studies are based on the assumption of fixed AP duplex modes\cite{Liu2023Performance,behdad2024multistatic,Shi2022Device}. In fact, the optimization of AP duplex mode represents a further management of the CLI within the system. By adjusting the AP duplex mode to meet the communication and sensing requirements of the system, the interference situation of the system can be further adjusted, thus improving the system performance. In previous studies, some of the literature has investigated the AP duplex mode optimization (ADMO) strategies in systems that only consider UL and DL communication without sensing\cite{Zhu2021Optimization,Mohammadi2023Network,Sun2023Hierarchical}. Zhu et al.\cite{ Zhu2021Optimization} investigated the ADMO for NAFD CF networks with the objective of maximizing the sum of SE, and proposed a parallel successive convex approximation-based algorithm and a reinforcement learning (RL) algorithm based on augmented Q-learning to solve this problem. Mohammadi et al.\cite{Mohammadi2023Network} proposed a joint optimization method to improve the sum of SE and energy efficiency of NAFD CF systems by jointly optimizing the AP duplex mode assignment, power control, and massive fading weights. Sun et al.\cite{Sun2023Hierarchical} proposed a pre-allocation mechanism for AP mode optimization in the NAFD CF networks, aiming to achieve a balance between SE and resource utilization.
	Liu et al.\cite{ liu2024cooperative } proposed three low-complexity heuristic algorithms to solve the ADMO problem in CF-ISAC networks, but each AP in their research architecture only supports sensing or communication. Elfiatoure et al.\cite{elfiatoure2023cellfree} proposed a greedy algorithm based on long term statistics for the ADMO in the CF-ISAC network. However, this study only considered the case of coexistence of DL communication and sensing within the system.
	
	Therefore, this paper focuses on the performance analysis and ADMO of multi-static ISAC system with UL and DL communication coexistence based on NAFD for CLI suppression. The main contributions of this paper are summarized as 
	\begin{enumerate}[]
		\item This paper considers a multi-static ISAC system with UL and DL communications coexisting and are able to respond rapidly to the UL and DL communication needs of UEs. To the best of the authors' knowledge, no previous studies have investigated in this kind of system.
		\item The CLI in the system is partially eliminated by using NAFD architecture. Then, we derive the closed-form expressions for the communication rate and location error rate (LER) to evaluate the performance of the communication and sensing operation, respectively. Based on the derived expressions, we establish a multi-objective optimization problem (MOOP) for getting the trade-off between communication and sensing.
		\item A deep Q-network (DQN)-based algorithm for ADMO is proposed, which enables the CPU to autonomously learn the near-optimal strategy for communication and sensing performance under different weights in this MOOP.
		\item The simulation results have verified the accuracy of the closed-form expression. By comparing the NAFD-based ISAC system with existing ISAC systems, the superiority of NAFD-based ISAC system in terms of communication performance and sensing performance has been demonstrated. Furthermore, the proposed ADMO algorithm can achieve performance close to that of the exhaustive method with extremely low complexity.
	\end{enumerate}
	\textbf{Notations}: Bold letters denote vectors or matrices. $\mathbf{I}_M$ denotes an M-dimensional identity matrix. The conjugate transpose and transpose are denoted by $(\cdot)^{\rm H}$ and $(\cdot)^{\rm T}$, respectively. $\left| \cdot \right|$ and $||\cdot ||$ represent the absolute value and spectral norm, respectively. $\mathbb{E}[\cdot]$ and $\rm cov[\cdot]$ denote expectation and covariance operators, respectively. $\rm tr(\cdot)$ denotes the trace of a matrix. The Kronecker product is denoted by $\otimes$. For a matrix $\mathbf{X}$, the $n$th element on the diagonal of $\mathbf{X}$ is represented by $[\mathbf{X}]_n$. Matrix inequality $\mathbf{X} \succeq \mathbf{Y}$ denotes that $\mathbf{X} - \mathbf{Y}$ is positive semidefinite. The estimation of $x$ is denoted by $\hat x$, and the estimation error is denoted by $e$. A circularly symmetric complex Gaussian random variable $x$ with mean zero and variance $\sigma^2$ is denoted as $x \sim \mathcal{CN} (0, \sigma^2)$. $\Gamma (k, \theta)$ denotes the Gamma distribution with parameters $k$ and $\theta$.
	\section{system model}
	Fig. \ref {fig:system_NAFD_ISAC} illustrates the schematic of a NAFD-based ISAC system that enables simultaneous UL and DL communication and sensing. The system consists of a CPU, $M$ uniformly distributed ISAC-APs, $K$ UEs, and $T$ passive sensing targets. All ISAC-APs are equipped with both communication and sensing capabilities. In this paper, we refer to these ISAC-APs as APs. $K_{\rm dl}$ UEs perform DL reception and $K_{\rm ul}$ perform UL transmission, and $K=K_{\rm dl}+K_{\rm ul }$. $\mathcal {M}$ is the set of all APs. $\mathcal {K}_{\rm dl}=\bigl\{{1,... ,K_{\rm dl}}\bigr\}$ and $\mathcal {K}_{\rm ul}=\bigl\{{1,... ,K_{\rm ul}}\bigr\}$ denote the set of all DL and UL UEs, respectively, and $\mathcal {K}=\mathcal {K}_{\rm dl}+\mathcal {K}_{\rm ul}$ denotes the set of all UEs. Each AP has $N$ antennas, and the UEs have single antennas. The APs can perform either UL reception or DL transmission, and this decision is made by the CPU. In the NAFD-based ISAC system, within a coherent block, the length of the coherent block is denoted as $\tau$, and the first $\tau_{\rm up}$ symbols are used for UL pilot training for the channel estimation of UEs. The next $\tau_{\rm dp}$ symbols are used for DL pilot training. In the DL training phase, the channel between the DL APs and DL UEs and the channel between the DL APs and UL APs can be estimated. Finally, the DL APs use coherent joint transmission to transmit the communication data and orthogonal sensing signal according to the specific beamforming and resource allocation algorithms designed by the CPU.
	\begin{figure}[htb]
		\centering
		\includegraphics[scale=0.6]{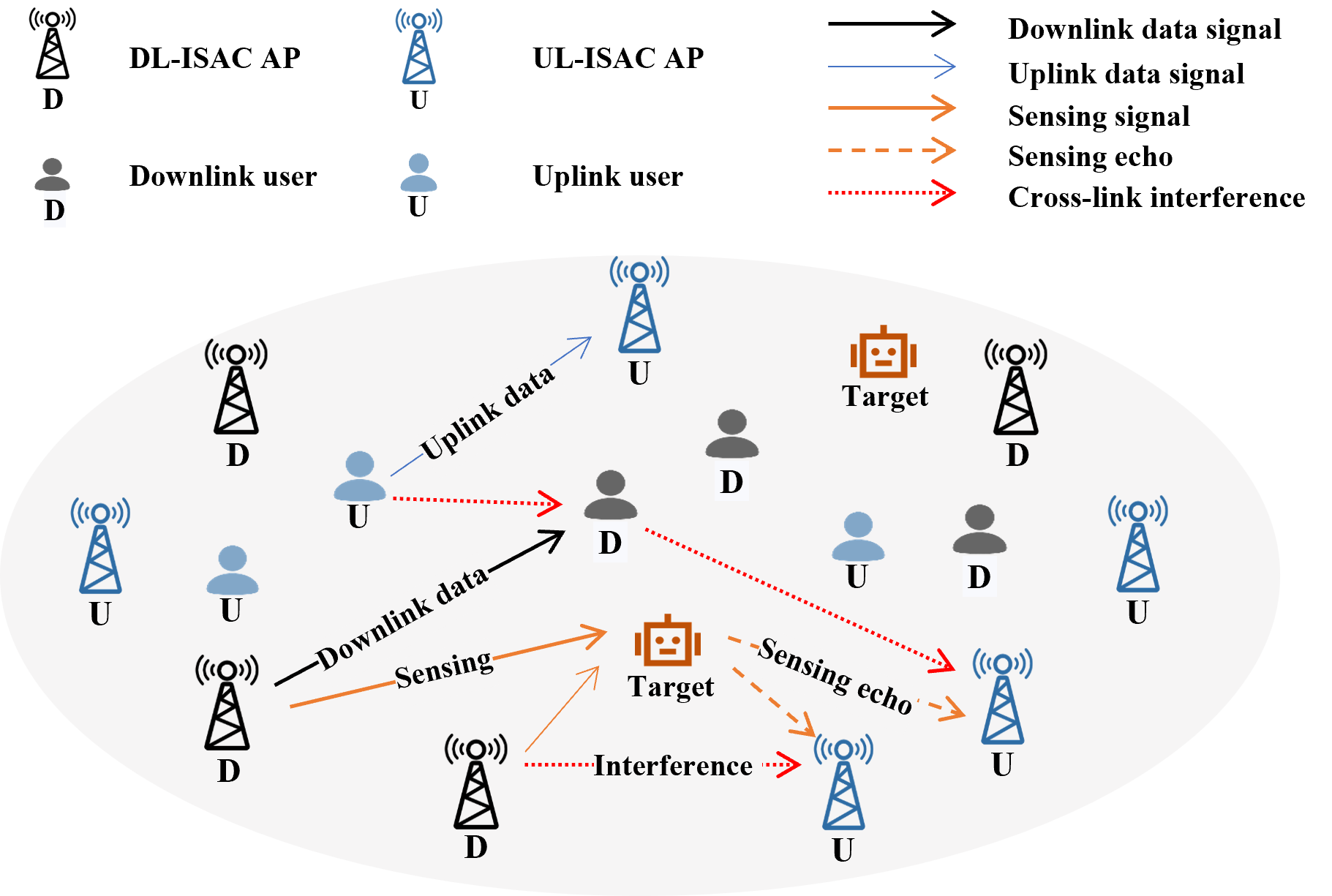}
		\caption {The schematic of the NAFD-based ISAC system}\label {fig:system_NAFD_ISAC}
	\end{figure}
	
	\subsection{Channel modeling}
	A quasi-static, flat-fading channel model is used in this study, where the channel remains static within each coherence interval and is flat in frequency. The channel is modeled as an aggregated channel containing a random scattering channel and the target-reflected channel. For instance, the random scattering channel between the $m$ AP and the $k$th UE is denoted as
	\begin{align}
		\label{form1:UA_channel}
		{\begin{aligned} {\mathbf{h}_{{\rm A}_m,{\rm U}_k}} = {\lambda _{{\rm A}_m,{\rm U}_k}}{\mathbf{q}_{{\rm A}_m,{\rm U}_k}} \in \mathbb{C}{^{N \times 1}},
		\end{aligned}}
	\end{align}
	where ${\lambda _{{\rm A}_m,{\rm U}_k}} \buildrel \Delta \over = { \rm d}_{{\rm A}_m,{\rm U}_k}^{ - {\alpha }}$  denotes the large-scale decay, $ {\rm d}_{{\rm A}_m,{\rm U}_k}^{}$ is the distance between the $m$th AP and the $k$th UE, ${ \rm \alpha}$ is the path loss exponent, ${\mathbf {q}_{{\rm A}_m,{\rm U}_k}}  \in \mathbb {C}{^{N \times 1}} \sim {\cal C}{\cal N}(0,{\mathbf {I}_N})$ is the small scale Rayleigh decay. For the target-reflected channel, it can be expressed mathematically as the product of the channel from AP to the targets and the channel from the targets to the UE.
	The target-reflected channel between the $m$th AP and the $t$th target can be modeled as
	\begin{align}
		\label{form2:AT_channel}
		{\begin{aligned} {\mathbf{h}_{{{ \rm A}_m},{{ \rm T}_t}}} = {\lambda _{ {{ \rm A}_m},{{ \rm T}_t}}}{\mathbf{q}_{ {{ \rm A}_m},{{ \rm T}_t}}} \in \mathbb{C}{^{N \times 1}},
		\end{aligned}}
	\end{align}
	where ${\lambda _{ {{{ \rm A}_m},{{ \rm T}_t}}}} \buildrel \Delta \over =  {\rm d}_{{ {{ \rm A}_m},{{ \rm T}_t}}}^{ - {\alpha }}$, $ {\rm d}_{{{{ \rm A}_m},{{ \rm T}_t}}}^{}$ is the distance between the $m$th AP and the $t$th target. Since the CPU knows the approximate position of the target but not the exact position, the position of the target can be expressed as ${\rm  d}_{{{ \rm A}_m},{{ \rm T}_t}}^{} \sim {\cal C}{\cal N}({\bar {\rm d}}_{{{ \rm A}_m},{{ \rm T}_t}}^{},\sigma _{{{ \rm A}_m},{{ \rm T}_t}}^2)$, $ {{\mathbf{q}}_{{{ \rm A}_m},{{ \rm T}_t}}} \sim {\cal C}{\cal N}({{\bar {\mathbf{ q}}}_{{{ \rm A}_m},{{ \rm T}_t}}},\chi _{{{ \rm A}_m},{{ \rm T}_t}}^2 {\mathbf {I}_N})$ is the steering vector between the $m$th AP and the $t$th target, where $ {\bar {\mathbf{q}}_{{{ \rm A}_m},{{ \rm T}_t}}}$ is denoted as
	\begin{align}
		\label{form3:qat}
		{\begin{aligned} {\bar {\mathbf{q}}_{{{ \rm A}_m},{{ \rm T}_t}}} = [1,{e^{j2\pi \frac{\mathbf{k}^{\rm T}_{mt} \mathbf{p}_{m1}}{\lambda }}}...,{e^{j2\pi \frac{\mathbf{k}^{\rm T}_{mt} \mathbf{p}_{mN}}{\lambda }}}] \in \mathbb{C} {^{N \times 1}},
		\end{aligned}}
	\end{align}
	where $\mathbf {k}_{mt}={[\cos ({\theta _{mt}}),\sin ({\theta _{mt}})]^{\rm T}}$ denotes the wave vector, and ${ \theta _{mt}}$ is the Direction of Arrival (DOA) of the $m$th AP to the $t$th target, $\mathbf {p}_{mi}={[{x_{mi}},{y_{mi}}]^{\rm T}}$ denotes the position of the $i$th antenna of the $m$th AP, and $\lambda$ denotes the signal wavelength. The channel between the $t$th target and the $k$th UE can be expressed as
	\begin{align}
		\label{form4:ut}
		{\begin{aligned} {{\rm h}_{{{ \rm T}_t},{ {\rm U}_k}}} = {\lambda _{{{ \rm T}_t},{ {\rm U}_k}}}{{ \rm q}_{{{ \rm T}_t},{ {\rm U}_k}}} \in \mathbb{C} {^{1 \times 1}},
		\end{aligned}}
	\end{align}
	where ${\lambda _{{{ \rm T}_t},{ {\rm U}_k}}} \buildrel \Delta \over = {\rm d}_{{{ \rm T}_t},{ {\rm U}_k}}^{ - {\alpha }}$, ${\rm d}_{{{ \rm T}_t},{ {\rm U}_k}}$ is the distance between between the $t$th target and the $k$th UE. Similarly the exact location of the target is unknown, ${\rm d}_{{{ \rm T}_t},{ {\rm U}_k}}^{} \sim {\cal C}{\cal N}({ \bar {\rm d}}_{{{ \rm T}_t},{ {\rm U}_k}}^{}$, ${{\rm q}_{{{ \rm T}_t},{ {\rm U}_k}}}  \sim {\cal C}{\cal N}(0,1) \in \mathbb {C}{^{1 \times 1}}$ is the small scale decay. Therefore, the aggregation channel between the $m$th AP and the $k$th UE can be expressed as
	\begin{align}
		\label{form:h_dl}
		{\begin{aligned}
				{\mathbf{h}_{mk}} \buildrel \Delta \over = {\mathbf{h}_{{{ \rm A}_m},{{ \rm U}_k}}} + \sum\nolimits_{t \in \mathcal{T}} {{\alpha _t}{{\mathbf{h}_{{{ \rm A}_m},{{ \rm T}_t}}}\mathbf{h}_{{{ \rm T}_t},{{ \rm U}_k}}}},
		\end{aligned}}
	\end{align}
	which consists of a target-independent channel and T NLOS paths reflecting off the target, where ${\alpha _t}$ is the reflection coefficient of the target $t$.
	The aggregation channel between APs can be represented as 
	\begin{align}
		{\begin{aligned}
				\mathbf{h}_{{\rm A},mn}\buildrel \Delta \over = {\mathbf{h}_{{ {\rm A}_m},{{ \rm A}_n}}} + \sum\nolimits_{t \in \mathcal{T}} {{\alpha _t}{\mathbf{h}_{{{ \rm A}_m},{{ \rm T}_t}}}{\mathbf{h}_{{{ \rm T}_t},{{ \rm A}_n}}}},\\
		\end{aligned}}
	\end{align}
	where $\mathbf{h}_{{ {\rm A}_m},{{ \rm A}_n}}={\lambda _{ {{ \rm A}_m}, {{ \rm A}_n}}}{\mathbf {q}_{ {{ \rm A}_m}, {{ \rm A}_n}}} \in \mathbb {C} {^{N \times N}}$denotes the random scattering channel between the $m$th AP and the $n$th AP, $ {\mathbf{h}_{{{ \rm T}_t}, {{ \rm A}_n}}} = {\lambda _{ {{ \rm T}_t}, {{ \rm A}_n}}}{\mathbf{q}_{  {{ \rm T}_t}, {{ \rm A}_n}}}$ denotes the channel between the $t$th target and the $n$th UL AP, $\lambda _{ {{ \rm T}_t}, {{ \rm A}_n}}=\lambda _{ {{ \rm A}_n}, {{ \rm T}_t}}$denotes the large-scale decay, ${{{\mathbf{q}}}_{{{ \rm T}_t},{{ \rm A}_n}}}$ denotes the steering vector between the $t$th target and the $n$th AP satisfying ${{\mathbf{q}}_{{{ \rm T}_t},{ {\rm A}_n}}} \sim {\cal C}{\cal N}({{\bar {\mathbf{q}}}_{{{ \rm T}_t},{{ \rm A}_n}}},\chi _{{\rm T}_t,{{{\rm A}_n}}}^2 {\mathbf {I}_N})$ where ${\bar {\mathbf{q}}_{{{ \rm T}_t},{{ \rm A}_n}}} = [1,{e^{j2\pi \frac{\mathbf{k}^{\rm T}_{tn} \mathbf{p}_{n1}}{\lambda }}},...,{e^{j2\pi \frac{\mathbf{k}^{\rm T}_{tn} \mathbf{p}_{nN}}{\lambda }}}] \in \mathbb{C} {^{1 \times N}}$,  $\mathbf {k}_{tn}={[\cos ({\phi_{nt}}),\sin ({\phi_{nt}})]^{\rm T}}$ denotes the wave vector, ${ {\phi} _{tn}}$ is the Direction of Departure (DOD) from the $t$th target and the $n$th AP.
	The aggregation channel between UEs can be represented as 
	\begin{align}
		{\begin{aligned}
				{{\rm h}_{{\rm I},u,l}} \buildrel \Delta \over = {{\rm h}_{{ {\rm U}_u},{{ \rm U}_l}}} + \sum\nolimits_{t \in \mathcal{T}} {{\alpha _t}{{\rm h}_{{{ \rm T}_t},{{ \rm U}_u}}}{{\rm h}_{{{ \rm T}_t},{{ \rm U}_l}}}},
		\end{aligned}}
	\end{align}
	where ${{\rm h}_{{ {\rm U}_u},{{ \rm U}_l}}}={\lambda _{ {{ \rm U}_k}, {{ \rm U}_u}}}{ {\rm q}_{ {{ \rm U}_k}, {{ \rm U}_u}}} \in \mathbb {C} {^{1 \times 1}}$ denotes the random scattering channel between UE $u$ and UE $k$.
	\subsection{Channel estimation}
	In this section, we perform the channel estimation for the UL and DL pilot training phase. In the UL pilot training phase, we estimate the channel between all UEs and APs for beamforming and data decoding. In the DL training phase, the channel between the DL APs and DL UEs and the channel between the DL APs and UL APs are estimated. Since all APs are connected to the CPU, the UL AP knows the signal sent by the DL AP. At the same time, it is assumed that the sensing signals sent by the DL APs are known to the DL UEs. Therefore, the CLI between the DL APs and UL APs and the CLI caused by the sensing signal to DL UEs can be partially eliminated by reconstructing the interfering signal when the channel is estimated in the DL pilot training phase. 
	\subsubsection{Autocorrelation Matrix/Coefficient}
	The aggregation channel differs from the conventional channel in that it considers the presence of the target to be sensed in the scene, resulting in changes to the correlation matrix or correlation coefficient of the channel. To facilitate the subsequent channel estimation, the following theorem provides the autocorrelation matrix/coefficient for each aggregated channel in the system.
	\begin{theorem}
		The autocorrelation matrix/coefficient of all channels present in the NAFD-based system can be expressed as
		\begin{align}
			\label{form:hmk_res}
			&\mathbb{E}\bigl[{{\mathbf{h}_{mk}}\mathbf{h}_{mk}^{\rm H}} \bigr]  = \lambda _{{{\rm A}_m},{{\rm U}_k}}^2{\mathbf{I}_N} + \sum\limits_{t \in \mathcal{T}}\alpha^2_t\bm{\zeta}_{{\rm A}_m,{\rm T}_t}\gamma_{{\rm T}_t,{\rm U}_k}={{\bm \phi}_{ml}},\\
			& \mathbb{E}\bigl[ {\mathbf{h}_{{\rm A},mn}^{}{{\mathbf{h}^{\rm H}_{{\rm A},mn}}}} \bigr]= \lambda _{{{\rm A}_m},{{\rm A}_n}}^2{\mathbf{I}_N} + \sum\limits_{t \in \mathcal{T}} {\alpha _t^2} \bm{\zeta}_{{\rm A}_m,{\rm T}_t}\bm{\zeta}_{{\rm T}_t,{\rm A}_n}= {\bm \phi}^{\rm A}_{mn},\\
			&\mathbb{E}\bigl[ {{{\bigl| {{{\rm h}_{{\rm I},u,l}}} \bigr|}^2}} \bigr] = \lambda _{{{\rm U}_u},{{\rm U}_l}}^2 + \sum\limits_{t \in \mathcal{T}} {\alpha _t^2}\gamma_{{\rm U}_u,{\rm T}_t}\gamma_{{\rm T}_t,{\rm U}_l}= {{\rm \phi}^{\rm u}_{ul}},\\
		\end{align}
		where $\bm{\zeta}_{a,b}=\bigl({{{\bar {\rm d}_{a,b}^{ -2{\alpha}}}} + \sigma _{a,b}^2}\bigr)\bigl( {{\bar {\mathbf q}}_{a,b}\bar{\mathbf q}_{a,b}^{\rm H} + \chi _{a,b}^2 {\mathbf {I}_N}} \bigr)$, $\gamma_{a,b}={{{\bar {\rm d}_{a,b}^{ -2{\alpha}}}} + \sigma _{a,b}^2}$, $a$ and $b$ represent the corresponding subscript types respectively.
	\end{theorem}
	\begin{proof}
		Please refer to Appendix \ref {theorem_1}.
	\end{proof}
	\subsubsection{Uplink Pilot Training}
	First, in the UL pilot training phase, the UE sends UL pilot to APs, the pilot sent by the UE $k$ is denoted as $\sqrt {{\tau _{\rm p}}} {\bm {\varphi} _k} \in \mathbb {C}{^{{\tau _{\rm p}} \times 1}}$ and ${\bigl\| {{\bm {\varphi} _k}} \bigr\|^2} = 1$, the pilot signal received at the $m$th AP is denoted as
	\begin{align}
		\label{form:uplink_pilot}
		{\begin{aligned} {\mathbf{Y}_{{\rm up},m}} = \sum\limits_{k \in \kappa } {\sqrt {{p_{\rm p}}{\tau _{\rm up}}} {\mathbf{h}_{mk}}} \bm{\varphi} _k^{\rm H} + {\mathbf{N}_{{\rm up},m}}.
		\end{aligned}}
	\end{align}
	The received signal can be projected on $\bm {\varphi}$ as
	\begin{align}
		\label{form:pilotmap}
		{\begin{aligned} {\mathbf{y}_{{\rm up},mk}} \buildrel \Delta \over = \sqrt {{p_{\rm p}}{\tau _{\rm up}}} {h_{mk}} + {\mathbf{N}_{{\rm up},m}}{{\bm \varphi} _k} \in \mathbb{C} {^{N \times 1}},
		\end{aligned}}
	\end{align}
	The MMSE channel estimation technique is employed for direct estimation of the aggregation channel, the estimated channel is denoted as $\hat {\mathbf {h}}_{mk}$, and the channel estimation error is denoted as ${{\mathbf {e}}_{mk}}$, which satisfies ${{\mathbf {h}}_{mk}}={\hat {\mathbf {h}}_{mk}}+{{\mathbf {e}}_{mk}}$. Then the autocorrelation matrix of the estimated channel is expressed as 
	\begin{align}
		\label{form:h_hat_h_hatmk}
		{\begin{aligned} \mathbb{E}\bigl[\hat {\mathbf{h}}_{mk}\hat {\mathbf{h}}^{\rm H}_{ mk}   \bigr]
				=\sqrt {{p_{\rm p}}{\tau _{\rm p}}}\mathbf{C}_{ mk}{\bm \phi}_{mk} \buildrel \Delta \over =\hat{\mathbf{R}}_{mk},
		\end{aligned}}
	\end{align}
	where $\mathbf{C}_{ mk}=\frac{\sqrt {{p_{\rm p}}{\tau _{\rm p}}}{\bm{ \phi}} _{mk}}{{p_{\rm p}}{\tau_{\rm p}}{\bm \phi}_{mk}+\mathbf{I}_{p}}$, and the autocorrelation matrix of the estimation error is expressed as $\mathbb{E}\bigl[{\mathbf{e}}_{ mk}{\mathbf{e}}^{\rm H}_{ mk}   \bigr]={\bm \phi}_{mk}-\hat{\mathbf{R}}_{mk} \buildrel \Delta \over ={{\bm \theta}_{mk}}$.
	
	\subsubsection{Downlink Pilot Training}
	During the DL pilot training phase, the AP sends the DL pilots. All the DL UEs can estimate the channel between the DL APs and DL UEs. The estimation method is similar to the UL pilot transmission stage, so it will not be described here. In addition, the UL APs can estimate the channel between the APs to eliminate the CLI. The UL APs transmit the received signal to the CPU through the backhaul link. The pilot sent by the DL AP $m$ is denoted as $\sqrt {{\tau _{\rm p}}} {\bm {\varphi} _m} \in \mathbb {C}{^{{\tau _{\rm p}} \times N}}$, and the signal received by the CPU can be expressed as 
	\begin{align}
		\label{form:downlink_pilot}
		{\begin{aligned}{{\mathbf{Y}_{\rm dp}}} = \sum\limits_{m \in {\mathcal{M} _{\rm }}} {\sum\limits_{n \in {\mathcal{M}_{\rm }} }{\sqrt {{p_{\rm p}}{\tau _{\rm dp}}}{\mathbf {g}}_{{\rm A},mn}\bm{\varphi}^{\rm H}_m} }  + {\mathbf{N}_{\rm dp}},
		\end{aligned}}
	\end{align}
	where ${\mathbf {g}}_{{\rm A},mn}\buildrel \Delta \over = {x_{{\rm d},m}}{x_{{\rm u},n}}{\mathbf {h}}_{{\rm A},mn}$ indicates the channel between DL AP $m$ and UL AP $n$, ${\mathbf {N}_{{\rm dp}}}\sim {\cal C}{\cal N}(0,\sigma _{{\rm dp}}^2\mathbf {I}) \in \mathbb {C}{^{N \times N}}$. Similarly, the estimated channel between the APs can be expressed as $\hat {\mathbf {g}}_{{\rm A},mn} $, satisfying ${\mathbf {g}}_{{\rm A},mn}=\hat{\mathbf {g}}_{{\rm A},mn} +{\mathbf {e}}_{{\rm A},mn}$, where ${\mathbf {e}}_{{\rm A},mn}$ denotes the channel estimation error. Then the autocorrelation matrix of the estimated channel is expressed as $\mathbb{E}\bigl[\hat {\mathbf{g}}_{{\rm A},mn}\hat {\mathbf{g}}^{\rm H}_{{\rm A},mn}   \bigr]
	=\sqrt {{p_{\rm p}}{\tau _{\rm p}}}\mathbf{C}^{\rm A}_{mn}{\bm \phi}_{mn} \buildrel \Delta \over =\hat{\mathbf{R}}^{\rm A}_{{mn}}$, where $\mathbf{C}^{\rm A}_{mn}=\frac{\sqrt {{p_{\rm p}}{\tau _{\rm p}}}{\bm{ \phi}}^{\rm A}_{mn}}{{p_{\rm p}}{\tau_{\rm p}}{\bm \phi}^{\rm A}_{mn}+\mathbf{I}_{p}}$, and the autocorrelation matrix of the estimation error is expressed as $\mathbb{E}\bigl[{\mathbf{e}}_{{\rm A},mn} {\mathbf{e}}^{\rm H}_{{\rm A},mn}   \bigr]={\bm \phi}^{\rm A}_{mn}-\hat{\mathbf{R}}^{\rm A}_{mk} \buildrel \Delta \over ={{\bm \theta}^{\rm A}_{mk}}$.
	\subsection{Communication Signal Model}
	To determine the operating modes of the APs, two binary assignment vectors, ${\mathbf{x}_{\rm u}},{\mathbf{x}_{\rm d}} \in {\bigl\{ {0,1} \bigr\}^{M \times 1}}$, are used to formulate the mode selection signal model. Specifically, if AP $i$($j$) is used for UL reception (DL transmission), ${x_{{\rm u},i}}\bigl( {{x_{{\rm d},j}}} \bigr)$ is assigned a value of 1, otherwise it is set to 0. In this paper, it is assumed that all antennas belonging to the same AP operate in the same mode and that each AP can only be designated as either UL or DL, i.e., ${x_{{\rm u},i}} + {x_{{\rm d},i}} = 1$.
	\subsubsection{Downlink Signal Model}
	For DL transmission, in each time slot, the DL AP sends a signal to the DL UE. The DL signal sent by the $m$th AP is denoted as
	\begin{align}
		\label{form:xm}
		{\begin{aligned} {\mathbf {x}_m} = \underbrace {\sum\limits_{k \in {\mathcal {K} _{dl}}} {\sqrt {{p_{{\rm dl},k}}} } \mathbf {w}_{mk}^{\rm c} s_{{\rm dl},k}^{}}_{\rm \mbox { communication signal}} + \underbrace {\sum\limits_{t \in \mathcal {T}} {\sqrt {{p_{{\rm s},t}}} \mathbf {w}_{mt}^{\rm s}{\varphi _{m,t}}} }_{\rm \mbox { sensing signal}} \in \mathbb {C} {^{N \times 1}},
		\end{aligned}}
	\end{align}
	where ${p_{{\rm dl},k}}$ denotes the communication power of the $k$th DL UE, $\mathbf {w}_{mk}^{\rm c}$ denotes the beamforming matrix of the communication signal, and $s_{{\rm dl},k}$ is the communication signal transmitted to the $k$th DL UE satisfying ${\mathbb E}\bigl [ {s_{{\rm dl},k}^{\rm H}} s_{{\rm dl},k}^{} \bigr] = 1$, ${{p_{{\rm s},t}}}$ denotes the power of the sensed signal assigned to the $t$th sensing target, $\mathbf {w}_{mt}^{\rm s}$ denotes the beamforming matrix of the sensed signal, and ${\varphi _{m,t}}$ denotes the sensing signal sent by the $m$th AP to perceive the $t$th target, satisfying ${\mathbb E}\bigl [ {{\varphi}_{m,t}^{\rm H}{\varphi} _{m,t}^{}} \bigr] = 1$. Then the signal received by the $l$th DL UE can be expressed as
	\begin{align}
		\label{form:y_dl}
		{\begin{aligned} y_l^{\rm dl} = \sum\limits_{m \in {\mathcal{M} _{}}} {{x_{{\rm d},m}}\mathbf{h}_{ml}^{\rm H}{\mathbf{x}_m}}  + \sum\limits_{u \in {\mathcal{K} _{ul}}} {\sqrt {{p_{{\rm ul},u}}} {{\rm h}_{{\rm I},u,l}}s_{{\rm ul},u}^{}}  + {n_{\rm dl}},
		\end{aligned}}
	\end{align}
	where $s_{{\rm ul},u}^{}$ denotes the UL data signal of the UL UE, ${p_{{\rm ul},u}}$ denotes the data transmission power of the $u$th UL UE, and $n_{{\rm dl}} \sim {\cal C}{\cal N}(0,\sigma^{2}_{\rm dl}) $ is the DL additive white Gaussian noise(AWGN). For DL communication, the CLI includes the interference of sensing signals sent by the DL ap to the DL UE, i.e., sensing-to-communication CLI and the interference of UL data signals to the DL UEs, i.e., UL-to-DL CLI. Assuming the sensing signal is known to the DL UEs, since the channel estimation has been performed during the pilot training phase, the sensing-to-communication CLI can be partially mitigated by reconstructing the sensing signal. Considering the channel estimation error, the remain received signal is expressed as
	\begin{align}
		\label{form:y_dl_e}
		{\begin{aligned} y_l^{\rm dl} = {\mathcal{D}_l} + \mathcal{N}_l^{\rm dl} + \mathcal{N}_l^{\rm error} + \mathcal{N}_l^{\rm CLI-s} + \mathcal{N}_l^{\rm CLI-c} + {n_{\rm dl}},
		\end{aligned}}
	\end{align}
	where ${\mathcal{D}_l} \buildrel \Delta \over = \sum\nolimits_{m \in {\mathcal{M} _{}}}^{} {\sqrt {{p_{{\rm dl},l}}} }\hat{\mathbf{h}}^{\rm H}_{{\rm dl},ml}\mathbf{w}_{ml}^{\rm c}s_{{\rm dl},l}$ denotes the desired effective DL signal, $\hat{\mathbf{h}}^{\rm }_{{\rm dl},ml}\buildrel\Delta \over = {x_{{\rm d},m}}\hat {\mathbf{h}}_{ml}^{\rm }$ denotes the DL estimated 
	channel between AP $m$ and UE $l$, $\mathcal{N}_l^{\rm dl} \buildrel \Delta \over = \sum\nolimits_{l' \in {\mathcal{K} _{\rm dl}}\backslash \{ l\} }^{} {\sum\nolimits_{m \in {\mathcal{M}  _{}}}^{} {\sqrt {{p_{{\rm dl},l'}}}\hat {\mathbf{h}}_{{\rm dl},ml}^{\rm H}\mathbf{w}_{ml'}^{\rm c}s_{{\rm dl},l}}}$ denotes the communication interference from other DL UEs to the $l$th UE, $\mathcal{N}_l^{\rm error} \buildrel \Delta \over = \sum\nolimits_{k \in {\mathcal{K}_{\rm dl}}} {\sum\nolimits_{m \in {\mathcal{M} _{}}}^{} {\sqrt {{p_{{\rm dl},k}}}\mathbf{e}_{{\rm dl},ml}^{\rm H}\mathbf{w}_{mk}^{\rm c}s_{{\rm dl},l}} }$ denotes the interference generated by the DL channel estimation error, $\mathbf{e}_{{\rm dl},ml}\buildrel \Delta \over ={x_{{\rm d},m}} \mathbf{e}_{ml}$ denotes the DL channel estimation error of the channel between AP $m$ and UE $l$, $\mathcal{N}_l^{\rm CLI-s} \buildrel \Delta \over = \sum\nolimits_{m \in \mathcal{M}} \sum\nolimits_{t \in \mathcal{T}} {\sqrt {{p_{{\rm s},t}}}{\mathbf{e}_{{\rm dl},ml}^{\rm H} \mathbf{w}_{mt}^{\rm s}{\varphi _{m,t}}} }$ denotes the residual sensing-to-communication interference, $\mathcal{N}_l^{\rm CLI-c} \buildrel \Delta \over = \sum\nolimits_{u \in {\mathcal{K} _{\rm ul}}} {\sqrt {{p_{{\rm ul},u}}} {{\rm h}_{{\rm I},u,l}}s_{{\rm ul},u}}$ denotes the UL-to-DL CLI.
	\subsubsection{Uplink Signal Model}
	For UL transmission, each UL UE sends a signal, while the DL APs send communication and sensing signals, so the signal received by the $n$th AP is denoted as
	\begin{align}
		\label{form:y_ul}
		{\begin{aligned}\mathbf{Y}_n^{\rm ul} = \sum\limits_{k \in {\mathcal{K} _{\rm ul}}} {{{x_{{\rm u},m}}\mathbf{h}_{mk}}s_{{\rm ul},k} + \sum\limits_{m \in {\mathcal{M} _{}}}{\mathbf {g}}_{{\rm A},mn}{\mathbf{x}_m}}  + \mathbf{n}_{{\rm ul},n},
		\end{aligned}}
	\end{align}
	where $s_{{\rm ul},k}$ denotes the UL data, and $\mathbf {n}_{{\rm ul},n} \sim \mathcal {C}\mathcal {N}(0,\sigma^{2}_{\rm ul}\mathbf {I})$ denotes the UL AWGN.
	The signals received from each AP are processed centrally at the CPU. For UL communication, the CLI includes the interference of the sensing and communication signals sent by the DL APs to the UL APs, i.e., DL-to-UL CLI. Since signals transmitted by the DL APs are familiar to the UL APs, and the channels between the DL APs and the UL APs are estimated during the DL pilot training phase, it is possible to reconstruct the signals transmitted by the DL APs. Therefore, the DL-to-UL CLI can be partially mitigated. Consequently, the signal of UE $u$ at the CPU can be mathematically represented as
	\begin{align}
		\label{form:r_u}
		{\begin{aligned}r_u^{\rm ul} = {\mathcal{D}_u} + \mathcal{N}_u^{\rm ul} + \mathcal{N}_u^{\rm error} + \mathcal{N}_u^{\rm CLI-c} + \mathcal{N}_u^{\rm CLI-s} + \mathcal{N}_u^{\rm noise},
		\end{aligned}}
	\end{align}
	where ${\mathcal{D}_u} \buildrel \Delta \over = \sum\nolimits_{n \in {\mathcal{M}}}^{}{\mathbf{v}_{nu}^{\rm H}{{\hat {\mathbf h}}_{{\rm ul},mu}}\sqrt {{p_{{\rm ul},u}}} }s_{{\rm ul},u}$ denotes the desired signal of the $u$th UE, ${{\hat {\mathbf h}}_{{\rm ul},mu}}\buildrel \Delta \over = {x_{{\rm u},n}}{{\hat {\mathbf h}}_{mu}}$ denotes the UL estimated channel between the $n$th AP $n$ and the $u$th UE, $\mathbf {v}_{nu}^{\rm H}$ denotes the receiver vector, $\mathcal{N}_u^{\rm ul} \buildrel \Delta \over =  \sum\nolimits_{u' \in {\mathcal{K}_{\rm ul}}\backslash \{ u\} }^{} {\sum\nolimits_{n \in {\mathcal{M} _{}}}^{} \sum\nolimits_{t \in \mathcal{T}}^{}{\mathbf{v}_{nu}^{\rm H}{{\hat {\mathbf h}}_{{\rm ul},mu}}\sqrt {{p_{{\rm ul},u'}}} } }s_{{\rm ul},u'}$ denotes the UL-interference caused by the other UL UEs, $\mathcal{N}_u^{\rm error} \buildrel \Delta \over =  \sum\nolimits_{k \in {\mathcal{K} _{\rm ul}}} {\sum\nolimits_{n \in {\mathcal{M}_{\rm }}}\sum\nolimits_{t \in \mathcal{T}}^{} {x_{{\rm u},n}}{\mathbf{v}_{nu}^{\rm H}{\mathbf{e}_{{\rm ul},nk}}\sqrt {{p_{{\rm ul},k}}}s_{{\rm ul},k}}}$ denotes the interference term due to the UL channel estimation error, ${\mathbf{e}_{{\rm ul},nk}}\buildrel \Delta \over ={{x_{{\rm u},n}}\mathbf{e}_{nk}}$ denotes the UL channel estimation error between the $n$th AP and the $k$th UL UE, $\mathcal{N}_u^{\rm CLI-c} \buildrel \Delta \over = \sum\nolimits_{m \in {\mathcal{M} _{}}}^{} {\sum\nolimits_{n \in {\mathcal{M} _{}}^{}}\sum\nolimits_{j \in {\mathcal{K} _{\rm dl}}}^{}}{\mathbf{v}_{nu}^{\rm H}{\mathbf{e}}_{{\rm A},mn}{\sqrt {{p_{{\rm dl},j}}} } \mathbf{w}_{mj}^{\rm c}s_{{\rm dl},j}}$ denotes the residual DL-to-UL CLI term of the communication signal, $\mathcal{N}_u^{\rm CLI-s} \buildrel \Delta \over =  \sum\nolimits_{m \in {\mathcal{M}}}^{}\sum\nolimits_{n \in {\mathcal{M}}}^{}\sum\nolimits_{t \in \mathcal{T}}^{}{\mathbf{v}_{nu}^{\rm H}{\mathbf{w}_{nt}^s\sqrt {{p_{{\rm s},t}}} {\varphi _{m,t}}}}$ denotes the residual DL-to-UL CLI term of the sensing signal, $\mathcal{N}_u^{\rm noise}\buildrel \Delta \over = \sum\nolimits_{n \in {\mathcal{M} _{}}}{x_{{\rm u},n}} {\mathbf{v}_{nu}^{\rm H}\mathbf{n}_n^{\rm ul}}$ denotes the product of UL Gaussian noise and receiver vector.
	\subsection{Sensing Signal Model}
	Analyzing from the sensing point of view, all DL APs send the communication signals together with the sensing signals. Typically, the orientation of the sensing beams are different from that of the communication beams to satisfy the sensing field of view, and the independent communication and sensing sequences provide a high correlation gain to differentiate the radar module whose goal is to estimate the position of the target. Moreover, in CF network, the APs serving a UE all transmit the data signals of that UE which leads to simultaneous same-frequency interference\cite{Mao2023An}. This will seriously limit the sensing performance if the communication signals are reused for sensing. Therefore, in this paper, we use the reflected echoes of the sensing beams to sense the unknown target. Therefore, the communication signals are interference to sensing.
	
		In the NAFD-based ISAC, all APs are connected to the CPU. As a result, the DL communication signals can be reconstructed by estimating the channels between the APs. This enables the signal received by the UL APs to be subtracted from the DL communication signals, while a similar operation can be performed for the UL communication signals. Assuming the UL communication data can be decoded correctly. Once the UL communication signals are extracted, their effect can be subtracted from the received signal. Therefore, following the above signal interference cancellation, the remaining signal at the AP $n$ can be expressed as
	\begin{align}
		\label{form:sense_rec}
		{\begin{aligned}{\mathbf {y}_n} =& {\sum\limits_{m \in \mathcal {M}}\sum\limits_{ t \in \mathcal {T}}^{}  {{\sqrt {{p_{{\rm s},t}}} }{\mathbf {g}}_{{\rm A},mn}\mathbf {w}_{mt}^{\rm s}\varphi _{m,t}^{}}}+ \bar{\mathbf{y}}_{{\rm dl},n}+\bar{\mathbf{y}}_{{\rm ul},n}+{\mathbf{n}_s},
		\end{aligned}}
	\end{align}
	where $\bar {\mathbf {y}}_{{\rm dl},n}=\sum\nolimits_{m \in {\mathcal{M} _{\rm }}}^{} \sum\nolimits_{j \in {\mathcal{K} _{\rm dl}}}^{}{ {{\mathbf{e}}_{{\rm A},mn}^{\rm }{\sqrt {{p_{{\rm dl},j}}} } \mathbf{w}_{mj}^{\rm c}}}$ denotes the DL communication interference cancellation residual term, $\bar {\mathbf {y}}_{{\rm ul},n}=\sum\nolimits_{k \in {\mathcal{K} _{\rm ul}}}  {\sqrt {{p_{{\rm ul},k}}} {\rm \mathbf{e}}_{{\rm ul},nk}}$ denotes the UL communication interference cancellation residual term, ${\mathbf{n}_s} \sim \mathcal{CN}(0,\sigma^2_s)$ denotes the sensing AWGN. Expanding the aggregated channel of the above equation yields the NLOS radial component reflected by the target and the target-independent multipath component, commonly known as clutter. Clutter is typically caused by scattering from permanent objects and temporary obstacles\cite{Liu2023Cluster}. The accurate modeling of clutter can be achieved through the long-term observation of the environment\cite{Rahman2020Framework}. Therefore, we assume that the clutter signal can be perfectly canceled, and the remaining received signal is denoted as
	\begin{align}
		\label{form:sense_h_nlos}
		{\begin{aligned}{\mathbf{y}_n} = {\sum\limits_{m \in \mathcal{M}_{}}}\sum\limits_{ t \in \mathcal{T}}^{} {{\sqrt {{p_{{\rm s},t}}} }\bar {\mathbf{g}}^{\rm}_{{\rm A},mn}\mathbf{w}_{mt}^{\rm s}\varphi _{m,t}^{}}+{\mathbf{z}},
		\end{aligned}}
	\end{align}
	where $\bar {\mathbf{g}}^{\rm }_{{\rm A},mn}\buildrel \Delta \over =\sum\nolimits_{ i \in \mathcal{T}}^{}{\alpha _i}{x_{{\rm d},m}}{x_{{\rm u},n}}{\mathbf{h}_{{{ \rm A}_m},{{\rm T}_i}}}{\mathbf{h}_{{{ \rm T}_i},{{ \rm A}_n}}}$ indicates the channel that the DL AP $m$ reflects through the targets to the UL AP $n$, $\mathbf{z}$ is the combination of sensing AWGN and residual UL and DL interference, modeled as a complex Gaussian distribution with covariance \cite{Sakhnini2021Bound}, and is denoted as
	\begin{align}
		\label{form:noise}
		{\rm cov}\bigl [{\mathbf z}\bigr]=\bigl[{\bigl\|{\bar {\mathbf { y}}}_{{\rm dl},n}\bigr\|^2+\bigl\|{\bar {\mathbf { y}}}_{{\rm ul},n}\bigr\|^2+\sigma^2_{{\rm s}}}\bigr]\mathbf{I}_{N\times1}. 
	\end{align}
	The remaining signals are individually mapped to their corresponding sensing symbols, and is denoted as
	\begin{align}
		\label{form:sense_h_nlos_single}
		{\begin{aligned}{\mathbf{y}_{m,n,t}} ={{\sqrt {{p_{{\rm s},t}}} }{\alpha _t}{\mathbf{h}_{{{ \rm A}_m},{{\rm T}_t}}}{\mathbf{h}_{{{ \rm T}_t},{{ \rm A}_n}}}\mathbf{w}_{mt}^{\rm s}}+\mathbf{z}\varphi _{m,t}^{\rm H}.
		\end{aligned}}
	\end{align}
	The sensing processing is performed based on the above formula.
	\section{Closed-form Expression Derivation}
	In this section, we derive the closed-form expressions for the communication rate and sensing performance of the system. We utilize the maximum ratio transmission (MRT) beamforming, i.e. $\mathbf{w}_{ml}^{\rm c}= \varepsilon _{l}^{}{{\hat {\mathbf{h}}_{ml}}}$, the maximal ratio combining (MRC) receiver, i.e. $\mathbf{v}_{nu}= \varepsilon _{u}^{}{{\hat {\mathbf{h}}_{nu}}},$ and conjugate-aware sensing beamforming\cite{demirhan2024cellfree}, i.e. $\mathbf{w}_{mt}^{\rm s}= \sqrt{{1}/{N}} {{\mathbf{q}_{{{ \rm A}_m},{{ \rm T}_t}}}}$ in the system, where $\varepsilon _{l}^{}$ and $\varepsilon _{u}^{}$ are the normalization coefficient.
	\subsection{Downlink Communication Rate}
	The DL communication rate of UE $l$ can be expressed as
	\begin{align}
		\label{form:rate_dl}
		{\begin{aligned} R_{l}^{\rm dl} = \bigl(1-\frac{\tau_{\rm dp}+\tau_{\rm up}}{\tau}\bigr)\mathbb{E}\bigl[ {\log_2 (1 + \gamma _{l}^{\rm dl})} \bigr],
		\end{aligned}}
	\end{align}
	where $\gamma^{\rm dl}_l$ represents the SINR of the signal received by the DL UE, which can be expressed as
	\begin{align}
		\label{form:sinr_dl}
		{\begin{aligned}\gamma _{l}^{\rm dl}=\frac{{ {{{\bigl| {{\mathcal{D}_l}} \bigr|}^2}}}}{{{{{\bigl| {\mathcal{N}_l^{\rm dl}} \bigr|}^2}} + {{{\bigl|  \mathcal{N}_l^{\rm error}  \bigr|}^2}} + {{{\bigl|\mathcal{N}_{l}^{\rm CLI-s}\bigr|}^2}} + {{{\bigl|\mathcal{N}_l^{\rm CLI-c} \bigr|}^2}}+ \sigma _{dl}^2}}.
		\end{aligned}}
	\end{align}
	As indicated by Eq. (\ref{form:rate_dl}), the DL communication rate formula involves an expectation, which can be resolved by incorporating the expectation into $\gamma _l$ and solving for each expectation separately. The derivation result is presented in the following theorem.
	\begin{theorem}
		\label{the:dl_rate}
		Considering the interference cancellation mechanism, the DL communication rate for the NAFD-based ISAC system with MRT beamforming and MRC receiver can be derived as a closed-form expression, given by
		\begin{align}
			\label{form:rate_dl_res}
			{\begin{aligned} R_{l}^{\rm dl} = \bigl(1-\frac{\tau_{\rm dp}+\tau_{\rm up}}{\tau}\bigr) {\log_2 (1 + \mathbb{E}\bigl[\gamma _l^{\rm dl}\bigr])},
			\end{aligned}}
		\end{align}
		where $\gamma _l$ represents the SINR of the signal received by the DL UE, which can be expressed as
		\begin{align}
			\label{form:r_dl_res_res} 
			{\begin{aligned}\mathbb{E}\bigl[\gamma _l^{\rm dl}\bigr]=\frac{{p_{\rm dl,l}}\varepsilon _l^{2}\bigl[{\sum\limits_{m \in {\mathcal{M} _{}}}^{} {x_{{\rm d},m}}{{{\hat k}_{ml}}\hat \theta _{ml}^2 + {{(\sum\limits_{m \in {\mathcal{M} _{\rm }}}^{} {x_{{\rm d},m}}{{{\hat k}_{ml}}\hat \theta _{ml}^{}} )}^2}} }\bigr]}{\mathcal{I}^{\rm inter}_{{\rm dl},l}+\mathcal{I}^{\rm e}_{{\rm dl},l}+\mathcal{I}^{\rm s}_{{\rm dl},l}+\mathcal{I}^{\rm CLI}_{{\rm dl},l}+\sigma^2_{\rm dl}},
			\end{aligned}}
		\end{align}
		where $\mathcal{I}^{\rm inter}_{{\rm dl},l}=\sum\nolimits_{l' \in {\mathcal{K} _{\rm dl}}\backslash \{ l\} }^{} \sum\nolimits_{m \in {\mathcal{M} _{}}}^{}p_{{\rm dl},l'}\varepsilon^2_{{{\rm dl},l'}}{{{\rm tr}({{{{ {\hat {\mathbf{R}}}}^{\rm dl}_{ml}}}}{{{{\hat {\mathbf{R}}}^{\rm dl}_{m{l'}}}}})}}$, $\mathcal{I}^{\rm error}_{{\rm dl},l}=\sum\nolimits_{k \in {\mathcal{K} _{dl}}}\sum\nolimits_{m \in {\mathcal{M} _{\rm }}}p_{{\rm dl},k}{\varepsilon^2_{{\rm dl},k}}{\rm tr}(\hat {\mathbf{R}}^{\rm dl}_{mk} {\bm{\theta}}^{\rm dl}_{ml})$, $\mathcal{I}^{\rm sense}_{{\rm dl},l}=\sum\nolimits_{t \in \mathcal{T}}\sum\nolimits_{m \in {\mathcal{M} _{\rm }}}\frac{{{p_{{\rm s},t}}}}{N}{\rm tr}({\bm{\psi} _{mt}}{\bm{\theta}}^{\rm dl}_{ml})$, $\mathcal{I}^{\rm CLI}_{{\rm dl},l}=\sum\nolimits_{u \in {\mathcal{K} _{\rm ul}}}{{p_{{\rm ul},u}}}{{\rm \phi}^{u}_{ul}}$.
	\end{theorem}
	\begin{proof}
		Please refer to Appendix \ref {theorem_2}.
	\end{proof}
	\subsection{Uplink Communication Rate}
	The communication rate of UL UE $u$ is as follows
	\begin{align}
		\label{form:rate_ul}
		{\begin{aligned}
				R_{u}^{\rm ul} = \bigl(1-\frac{\tau_{\rm dp}+\tau_{\rm up}}{\tau}\bigr)\mathbb{E}\bigl[ {\log_2 (1 + \gamma _{u}^{\rm ul })} \bigr]
		\end{aligned}},
	\end{align}
	where $\gamma _{u}^{}$ is the SINR of the UL UE $u$, expressed as
	\begin{align}
		\label{form:sinr_ul}
		{\begin{aligned}\gamma _{u}^{\rm ul }=\frac{\bigl|{\mathcal{D}_u} \bigr|^2}{\bigl|\mathcal{N}_u^{\rm ul}\bigr|^2 +\bigl| \mathcal{N}_u^{\rm e}\bigr|^2 + \bigl|\mathcal{N}_u^{\rm CLI-c}\bigr|^2 + \bigl|\mathcal{N}_u^{\rm CLI-s}\bigr|^2 + \bigl| \mathcal{N}_u^{\rm n}\bigr|^2}.
		\end{aligned}}
	\end{align}
	Similarly, by incorporating the expectation into $\gamma _{u}^{\rm ul }$ and solving each expectation separately, we can derive the closed-form expression for the UL communication rate of the system. The results of the derivation are given in the following theorem.
	\begin{theorem}
		\label{the:ul_rate}
		Considering the interference cancellation mechanism, the UL communication rate for the NAFD-based ISAC system with MRT beamforming and MRC receiver can be derived as a closed-form expression, given by
		\begin{align}
			\label{form:rate_ul_res}
			{\begin{aligned}
					R_{u}^{\rm ul} = \bigl(1-\frac{\tau_{\rm dp}+\tau_{\rm up}}{\tau}\bigr) {\log_2 (1 + \mathbb{E}\bigl[\gamma _u^{\rm ul}\bigr])} 
			\end{aligned}},
		\end{align}
		where $\gamma _{u}^{}$ is the SINR of the UL UE $u$, expressed as
		\begin{align}
			\label{form:sinr_ul_res}
			{\begin{aligned}\mathbb{E}\bigl[\gamma _u^{\rm ul}\bigr]=\frac{{{p_{{\rm ul},u}}}\varepsilon _{{\rm }u}^{2}\bigl[{\sum\limits_{n \in {\mathcal{M} _{\rm }}}^{} {x_{{\rm u},n}}{ \hat k_{nu}^{} \hat \theta _{nu}^2 + } (\sum\limits_{n \in {\mathcal{M}_{}}}^{} { {x_{{\rm u},n}}\hat k_{nu}^{} \hat \theta _{nu}^{}})^2}\bigr]}{\mathcal{I}^{\rm inter}_{{\rm ul},u}+\mathcal{I}^{\rm e}_{{\rm ul},u}+\mathcal{I}^{\rm CLI-c}_{{\rm ul},u}+\mathcal{I}^{\rm CLI-s}_{{\rm ul},u}+\mathcal{I}^{\rm n}_{{\rm ul},u}},
			\end{aligned}}
		\end{align}
		where $\mathcal{I}^{\rm inter}_{{\rm ul},u}=\sum\nolimits_{u' \in {\mathcal{K}_{\rm ul}}\backslash \{ u\} }^{}\sum\nolimits_{n \in {\mathcal{M} _{\rm }}}{p_{{\rm ul},u}}\varepsilon _{{\rm ul},u}^{2}{{{\rm tr}({{{{\hat {\mathbf{R}}}^{\rm ul}_{nu}}}}{{{{\hat {\mathbf{R}}}^{\rm ul}_{n{u'}}}}})}}$, $\mathcal{I}^{\rm CLI-c}_{{\rm ul},u}=\sum\nolimits_{m \in {\mathcal{M} _{\rm }}}^{}\sum\nolimits_{n \in {\mathcal{M} _{\rm }}}^{} \sum\nolimits_{j \in {\mathcal{K} _{\rm dl}}}^{}\rho_{{\rm c},uj}{\rm tr}(\bm{\theta}^{\rm A}_{mn}\hat {\mathbf{R}}^{\rm ul}_{nu}\hat {\mathbf{R}}^{\rm dl}_{mj})$, $\mathcal{I}^{\rm CLI-s}_{{\rm ul},u}=\sum\nolimits_{m \in {\mathcal{M} _{\rm }}}^{}\sum\nolimits_{n \in {\mathcal{M} _{\rm }}}^{}\sum\nolimits_{t \in \mathcal{T}}{p_{{\rm ul},u}}\frac{\varepsilon_{{\rm ul},u}^{2}}{N}{{\rm tr}\bigl({\bm{\theta}_{mn}^{\rm A}\hat{\mathbf{R}}^{\rm ul}_{nu}\bm{\psi}_{mt}}\bigr)}$, $\mathcal{I}^{\rm error}_{{\rm ul},u}=\sum\nolimits_{k \in {\mathcal{K} _{\rm ul}}} {\sum\nolimits_{n \in {\mathcal{M}_{\rm }}}}{p_{{\rm ul},k}}\varepsilon_{{\rm ul},u}^{2}{\rm tr}(\hat {\mathbf{R}}^{\rm ul }_{nu}\hat {\bm{\theta}}^{\rm ul}_{nk})$, $\mathcal{I}^{\rm noise}_{{\rm ul},u}=\sigma^2_{\rm ul}\varepsilon_{{\rm ul},u}^{2}\sum\nolimits_{n \in {\mathcal{M} _{\rm}}}{\rm tr}\bigl(\hat{\mathbf{R}}^{\rm ul}_{nu}\bigr)$.
	\end{theorem}
	\begin{proof}
		Please refer to Appendix \ref {theorem_3}.
	\end{proof}
	\subsection{Sensing Localization Estimation Rate }
	  To standardize the performance evaluation of the ISAC system, we adopt the LER, which is analogous to the communication rate, to evaluate the sensing performance of the system. Assuming that the sensing targets are spatially separated, the LER of the system can be expressed as the sum of the LER of each target \cite{Guo2021Performance}. Specifically, the LER of the target $t$ can be obtained by solving the following
	\begin{align}
		\label{form:sense_r}
		{\begin{aligned}R^{\rm est}_{t}=\bigl(1-\frac{\tau_{\rm dp}+\tau_{\rm up}}{\tau}\bigr)\mathbb{E}\bigl[{\log _2}(1 + \frac{{\sigma _{loc}^2}}{{\rm CRLB}_{{\rm loc},t}})\bigr],
		\end{aligned}}
	\end{align}
	where $\sigma _{loc}^2$ is the uncertainty of the target location and ${\rm CRLB}_{{\rm loc},t}$ is the CRLB for location estimation. Typically, the estimation of location is decomposed into the estimation of the angle of the direction of arrival $\theta_{n,mt}$, the angle of the direction of departure ${\phi _{m,tn}}$, and the distance $d$, respectively. Thus, the CRLB is also decomposed into
	\begin{align}
		\label{form:crlb}
		{\begin{aligned}{\rm CRLB}_{{\rm loc},t}=\sum\limits_{m \in \mathcal{M}_{\rm dl}}\sum\limits_{n \in \mathcal{M}_{\rm ul}}\frac{\sigma^{2}_{d_{m,n,t}}+\sigma^{2}_{\theta_{n,mt}}+\sigma^{2}_{\phi_{m,tn}}}{M_{\rm dl}M_{\rm ul}},
		\end{aligned}}
	\end{align}
	where $\sigma^{2}_{d_{m,n,t}}$ denotes the distance CRLB of target $t$ sensing based on $\mathbf{y}_{m,n,t}$ in Eq. (\ref{form:sense_h_nlos_single}),  $\sigma^{2}_{\theta_{n,mt}}$ denotes the CRLB of DOA between DL AP $m$ and target $t$ sensing based on $\mathbf{y}_{m,n,t}$ , $\sigma^{2}_{\phi_{m,tn}}$ denotes the CRLB of DOD between target $t$ and UL AP $n$ sensing based on $\mathbf{y}_{m,n,t}$. The CRLB closed-form expression for each estimated parameter is given by the following theorem.
	\begin{theorem}
		\label{the:LER}
		In the case of large antenna arrays, the CRLB closed-form expression for each estimated parameter in the NAFD-based ISAC system is given by
		\begin{align}
			&\sigma^{2}_{d_{m,n,t}}= \frac{\sigma^2_{{\rm dl},n}+\sigma^2_{{\rm ul},n}+\sigma^2_{\rm s}}{{\alpha}_tp_{{\rm s},t}{\pi ^2}{{({\Delta f}/{c})^2}}\eta _{m,n,t}^2{N^2}{{{\bigl\| {\mathbf{w}_{m,t}^{\rm s}} \bigr\|}^2}}}
			\\
			&\sigma^{2}_{\theta_{n,mt}}= \frac{{{\lambda ^2}(\sigma^2_{{\rm dl},n}+\sigma^2_{{\rm ul},n}+\sigma^2_{\rm s})}}{{{4{\alpha}_tp_{{\rm s},t}{\pi ^2}\eta _{m,n,t}^2(N{B_{m,t}} - A_{m,t}^2){{{\bigl\| {\mathbf{w}_{m,t}^{\rm s}} \bigr\|}^2}}} }},\\ 	
			&\sigma^{2}_{\phi_{m,tn}}= \frac{{{\lambda ^2}(\sigma^2_{{\rm dl},n}+\sigma^2_{{\rm ul},n}+\sigma^2_{\rm s})}}{{{4{\alpha}_tp_{{\rm s},t}{\pi ^2}\eta _{m,n,t}^2(N{B_{t,n}} - A_{t,n}^2){{{\bigl\| {\mathbf{w}_{m,t}^{\rm s}} \bigr\|}^2}}} }},
		\end{align}
		where ${A_{m,t}} = \sum\nolimits_{i = 1}^N {\bigl[{y_{mi}}\cos (\theta _{mt})- {x_{mi}}\sin (\theta _{mt})\bigr]}$, ${B_m} = {\sum\nolimits_{i = 1}^N {\bigl[{y_{mi}}\cos ({\theta _{mt}})- {x_{mi}}\sin (\theta _{mt})\bigr]} ^2}$, ${A_{m,t}}$ and ${B_{t,n}}$ are related to the antenna array on the transmitter, ${A_{t,n}} = \sum\nolimits_{i = 1}^N {\bigl({y_{ni}}\cos ({\phi _{tn}})- {x_{ni}}\sin ({\phi _{tn}})\bigr)}$, ${B_{t,n}} = {\sum\nolimits_{i = 1}^N {\bigl({y_{ni}}\cos ({\phi _{tn}})- {x_{ni}}\sin ({\phi _{tn}})\bigr)} ^2}$, ${A_{t,n}}$ and ${B_{t,n}}$ are related to the antenna array on the transmitter, $\sigma^2_{{\rm dl},n}=\sum\nolimits_{i \in {\mathcal{M} _{\rm }}}^{}\sum\nolimits_{j \in {\mathcal{K} _{\rm dl}}}^{}{{{p_{{\rm dl},j}}} }\varepsilon _{{\rm dl},j}^{2}{\rm tr}\bigl(\bm {\theta}_{mn}^{\rm A}\hat {\mathbf{R}}^{\rm dl}_{ij}\bigr)$ denotes the DL communication signal interference residual power at the UL AP $n$, $\sigma^2_{{\rm ul},n}=\sum\nolimits_{k \in {\mathcal{K} _{\rm ul}}} {{p_{{\rm ul},k}}}{\rm tr}\bigl(\hat{\mathbf{R}}^{\rm ul}_{nk}\bigr)$ denotes the UL communication signal interference residual term power at the UL AP $n$.
	\end{theorem}
	\begin{proof}
		Please refer to Appendix \ref {theorem_4}.
	\end{proof}
	\section{AP Duplex Mode Optimization}
	\subsection{Problem Formulation}
	\begin{figure}
		\centering
		\subfloat {\includegraphics [ scale=0.3]{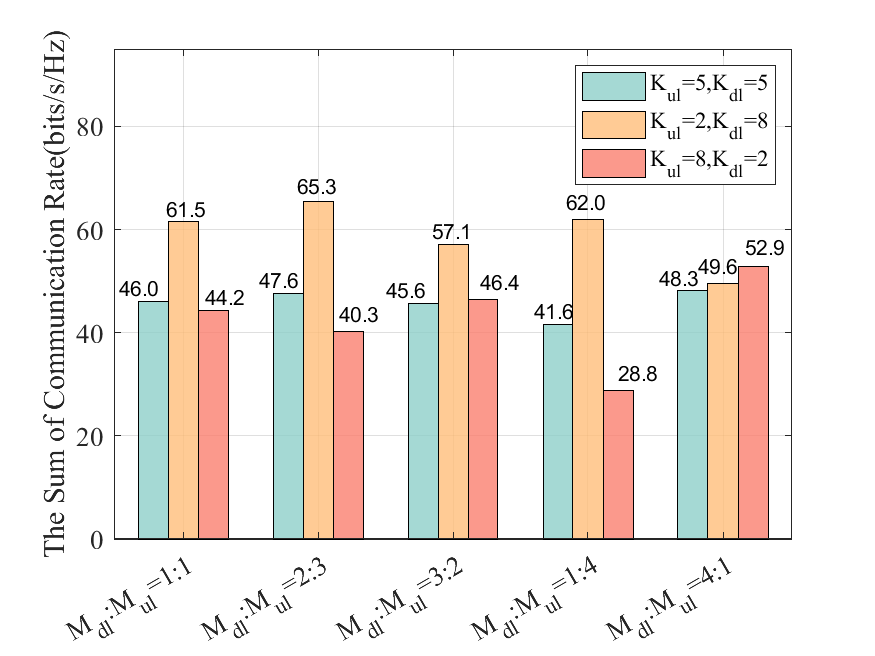}}
		\hspace{-2mm}
		\subfloat {\includegraphics [ scale=0.3]{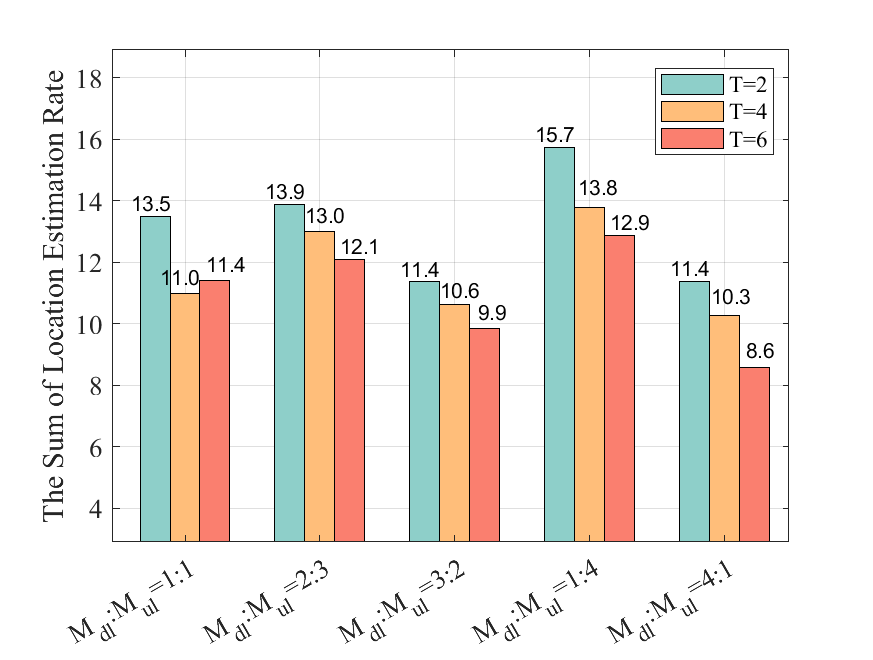}}
		\caption {Changes of system communication and sensing performance in different scenarios ($M=10$, $K=10$). (a) Changes in system communication performance. (b) Changes in system sensing performance.}
		\label{fig:rate_com}
	\end{figure}
	Fig.\ref{fig:rate_com} shows the simulation results for communication and sensing performance with a fixed total number of APs and varying numbers of sensing UEs and targets, as well as different UL and DL ratios of the APs. The results demonstrate that the duplex mode selection of APs has a significant impact on both the communication and sensing performance of the system.
	Based on the communication performance expression derived, it is evident that the communication rates of both DL and UL are impacted by the CLI between the AP and the UEs. To mitigate the CLI, the CPU can operate APs with different duplex modes by spatially isolating them. However, it should be noted that the sensing performance is also linked to the distance from the DL APs to the UL APs, as revealed by the derived CRLB expression in Theorem \ref{theorem_4}. The use of spatially isolated duplex modes may lead to lower sensing signal energy received by the UL APs, which in turn affects the LER. Thus, it is imperative to strike a balance between communication performance and sensing performance. 
	
	In this section, we use the theory of multi-objective optimization to propose the following MOOP. The first optimization objective is to maximize the sum communication rate, denoted as
	\begin{align} 
		\mathcal{P}_1:~\underset{\mathbf{x}_{\rm u},\mathbf{x}_{\rm d}}{\mathop{\max }}&~{f_1}=\sum\limits_{l \in {\mathcal{K} _{\rm dl}}}{R_{l}^{{\rm dl}} + } \sum\limits_{u \in {\mathcal{K} _{\rm ul}}} {R_{u}^{{\rm ul}}},\\
		\text{s.t.\;\;}&\mathbf{x}_{\rm u}+\mathbf{x}_{\rm d}=\mathbf{I}_M.\label{form:mode_select}
	\end{align}
	Each AP can be either UL or DL at any time, so the duplex mode selection vector satisfies the constraint in Eq. (\ref {form:mode_select}).
	
	Mathematically, the optimization problem for maximizing LER can be expressed as
	\begin{align}
		\mathcal{P}_2:~\underset{\mathbf{x}_{\rm u},\mathbf{x}_{\rm d}}{\mathop{\max }}&~{f_2} = \sum\limits_{t \in {\mathcal{T}}}{R_{{\rm loc},t}^{\rm est}},\\
		\text{s.t.\;\;}&(\ref{form:mode_select}).
	\end{align}
	Building on the previous analysis, given the conflict between the two objectives, we propose a MOOP to investigate the trade-off between them. This optimization problem can be mathematically formulated as
	\begin{align}
		\mathcal{P}_3:~\underset{\mathbf{x}_{\rm u},\mathbf{x}_{\rm d}}{\mathop{\max }}&~\mathbf{f} = {\bigl[ {{f_1},{f_2}} \bigr]^T},\\
		\text{s.t.\;\;}&(\ref{form:mode_select}).
	\end{align}
	The objective of $\mathcal{P}_3$ is to optimize the combined performance of communication and sensing. To tackle this MOOP, the concept of Pareto optimality can be utilized. An allocation policy is considered Pareto optimal if no other policy exists that can enhance one objective without compromising others. Therefore, a MOOP may have multiple Pareto optimal solutions. While these solutions are selected based on their dominance relationships, a closer examination of these optimal solutions reveals that they correspond to distinct multi-objective weights. By adjusting these weights, it is feasible to regulate the tradeoff between different objectives and investigate various regions of the Pareto boundary. We aim to develop an algorithm that can efficiently identify a Pareto-optimal solution for $\mathcal{P}_3$.
	\subsection{DQN-based Solution}
	The optimization problem discussed above is classified as NP-hard, indicating its high computational complexity.  To address this challenge, Reinforcement Learning (RL) has emerged as a promising approach that involves an intelligent agent perceiving the state of the environment and making decisions based on feedback to maximize rewards.  The DQN algorithm is a type of RL that combines deep learning and Q-learning\cite{Arulkumaran2017Deep}. DQN utilizes a neural network to learn Q-value functions, enabling it to handle complex environments. 
	Subsequently, we shall present the implementation of the DQN algorithm in the context of ADMO within the NAFD-based ISAC system. This approach involves four key components: 
	\begin{enumerate}[]
		\item \textit{Agent:} A centralized agent was established on the CPU side to facilitate the selection of duplex mode for each AP by monitoring the current state of the environment.
		\item \textit{State:} The current state of the agent includes a set of $M$ APs, denoted as $\bigl\{ {{x_1},{x_2},...,{x_M}} \bigr\}$. The state of each AP is represented in binary format, where $x_m$ of 0(1) indicates that the AP $m$ is operating in the UL(DL) mode.
		\item \textit{Action:} As there are only two operating modes for each AP, the system can take actions to change the original operating mode, either from UL to DL or from DL to UL. Consequently, these actions can be represented as a Markov action space denoted as $A= \bigl\{a_1, a_2,..., a_M\bigr\}$ where $a_t$ represents the action of switching the operating mode of AP $m$ from its original mode to the other mode.
		\item \textit{Reward:} Define the reward as
		\begin{align}
			\label{reward}
			{\begin{aligned} 
					{r_t} \buildrel \Delta \over = {\omega _c}{f_1} + {\omega _s}{f_2},
			\end{aligned}}
		\end{align}
		where ${\omega _c}$ and ${\omega _s}$ denote the weights of the two targets respectively. The working mode of each AP can be switched between UL mode and DL mode by the policy adopted by the agent.
	\end{enumerate}
	The DQN algorithm enhances its efficiency and applicability by adopting a value function-based approach instead of directly storing Q-values. To improve the training stability, the Double DQN structure is utilized in this section, which consists of two deep neural networks(DNN): A Q-evaluate network and a Q-target network. The Q-evaluate network is responsible for selecting actions, while the Q-target network computes the Q-value of the target. To minimize the correlation between the estimate of the Q-value of the target and the actual value, and to mitigate the over-estimation caused by the correlation, the parameters of the Q-target network are periodically copied from the Q-evaluate network. The Q-value update strategy for DQN is
	\begin{align}
		\label{form:Qvalue_DQN}
		{\begin{aligned} 
				Q({s_t},{a_t};\theta) ={r_t} + \gamma *\mathop {\max }\limits_{a_{t+1}}~Q({s_{t+1}},{a_{t+1}};\theta),
		\end{aligned}}
	\end{align}
	where $Q ({s_t},{a_t};\theta)$ is the Q value corresponding to taking action $a_t$ in state $s_t$ when the Q-evaluate network parameter is $\theta$. The DQN algorithm framework is shown in Fig.\ref {fig:DQN}.
	\begin{figure}[htb]
		\centering
		\includegraphics[scale=0.7]{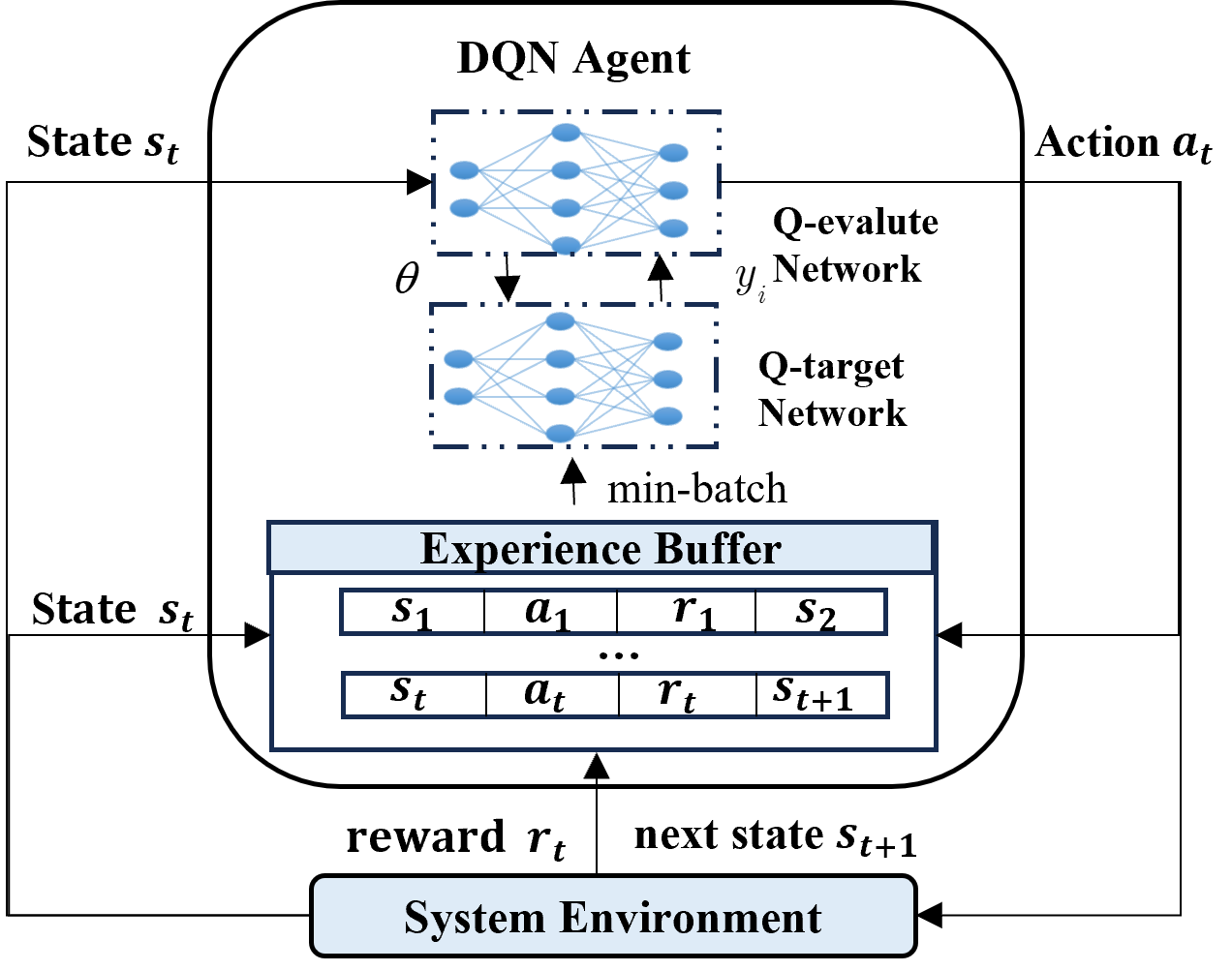}
		\caption {Diagram of the DQN algorithm}
		\label{fig:DQN}
	\end{figure}
	
	First, the agent observes the current state $s_t$ and inputs $s_t$ into the Q-evaluate network. The Q-evaluate network outputs $Q ({s_t},{a_m};\theta)$ for all possible actions $a_m$ in $s_t$. The action with the highest Q value, $a_t=\max~Q ({s_t},{a_m};\theta)$ is selected as the next action. After obtaining the action $a_t$, the reward $r_t$ can be calculated using Eq. (\ref{reward}). During training, the current action $a_t$, the state $s_t$, the reward $r_t$, and the next state $s_{t+1}$ are stored as a set of experiences $e=\bigl \{a_t,s_t,r_t,s_{t+1}\bigr\}$ in the experience replay buffer $\mathcal {D}_{\rm replay}$. After a sufficient number of experiences have been stored, a set of experience values $\mathcal {D}_{\rm min}$ is randomly sampled from $\mathcal {D}_{\rm replay}$ for updating the parameters $\theta$ in the Q-evaluate network. The DQN algorithm adopts experience replay to prevent sample correlation and overfitting issues caused by selecting a specific set of experiences. The network parameters are updated using the gradient descent method
	\begin{align}
		\label{form:theta}
		{\begin{aligned} 
				\theta  = \theta  - \alpha *{\nabla _\theta }L(\theta),
		\end{aligned}}
	\end{align}
	where $\alpha$ is the learning rate and $L (\theta)$is the loss function, which is calculated by
	\begin{align}
		\label{form:Loss}
		{\begin{aligned} 
				L(\theta) = \mathbb{E}_{e_i \sim \mathcal{D}}\bigl[({{y_i} - Q({s_i},{a_i};\theta )})^2\bigr],
		\end{aligned}}
	\end{align}
	where ${y_i} = {r_i} + \gamma *\mathop {\max }\limits_{a_{t+1}} Q(s_{i+1},a_{i+1};\theta_{\rm target})$
	, $\theta_{\rm target}$ refers to the parameters of the Q-target network. When selecting actions, DQN typically employs a $\varepsilon$-greedy strategy, which selects a random action with a certain probability $\varepsilon$ to explore the environment, and selects the action with the highest predicted Q value with a probability of 1-$\varepsilon$ to exploit existing information. In the context of the DQN algorithm, two important concepts are introduced, namely the episode and the time step $t_{\rm max}$. An episode is a sequence of actions and state transitions that start from the initial state of the environment and continue until the intelligent agent reaches the termination state. The behaviors of the agent between Episodes are considered independent, and the agent can learn and improve its behaviors through each episode. $t_{\rm max}$ is typically used to represent the number of time steps accumulated before a gradient update is performed. It is used to control the length of the experience sequence stored and sampled in the experience replay, as well as to determine when to update the Q-target network. The pseudocode of the algorithm can be found in Algorithm \ref{algo:DQN}.
	\begin{algorithm}
		\caption {DQN-based ADMO algorithm}
		\label{algo:DQN}
		\begin{algorithmic}[1]
			\State \textbf {Initialization:} channel and precoding vectors $\mathbf {h}_{mk},\mathbf {h}_{{\rm I},u,l},\mathbf {h}^{\rm AP}_{{\rm A},mn}$, network parameters $\theta$ and $\theta_{ target}$.
			\While {$episode<E_{\rm max}$}
			\State Initialization environment and get state $s_0$.
			\While {$t < t_{\rm max}$}
			\State $\varepsilon$-greedy algorithm to select action $a_t$.
			\State stores $\bigl\{{{s_t},{a_t},{r_t},{s_{t+1}}} \bigr\}$ into $D_{\rm replay}$.
			\State Sample a min-batch $D_{\rm min}$ from $D_{\rm replay}$.
			\For {each experience sample $\bigl\{ {{s_i},{a_i},{r_i},{s_{i+1}}} \bigr\}$}
			\State Calculate Q-value by Eq. (\ref{form:Qvalue_DQN}).
			\State Update $\theta$ by Eq. (\ref{form:Loss}).
			\EndFor
			\If{$t \bmod t_{update}==0$}
			\State Update $\varepsilon$, $\theta_{\rm target}$.
			\EndIf
			\EndWhile
			\EndWhile
			\State \textbf {Output:} The action corresponds to the optimal state as well as the reward.
		\end{algorithmic}
	\end{algorithm}
	\subsection{Algorithm Complexity Analysis}
	The exhaustion method exhibits an exponential growth in computational complexity. In contrast, the computational complexity of DQN mainly arises from matrix multiplication in the forward and backward propagation of the deep neural network and the number of training iterations. It can be approximated as $O (2 E_{\rm max}{t_{\max }}{D_{\rm min}}(M {n_0} + \sum\nolimits^{m-2}_{i=1} n_{i} n_{i+1}+ 2M {n_m}))$, where $M$ is the number of AP, the input layer of the neural network is $M$, the output layer is $2M$, $D_{\rm min}$ represents the size of the min-batch, $t_{\max }$ represents the number of iterations, and $m$ represents the number of hidden layers in the neural network, with $n_{i}$ denoting the number of neurons in the corresponding hidden layer. Notably, the computational complexity of the DQN is linearly related to $M$, resulting in lower complexity and better scalability.
	\section{Simulation Result}
	In this section, we present a comparative analysis of the proposed NAFD-based ISAC system with the existing ISAC system design. Furthermore, the accuracy of the closed-form expression derived in the previous section is verified via Monte Carlo simulation. Finally, based on the closed-form expression results, a comparative evaluation of the proposed DQN-based algorithm with four commonly used algorithms in ADMO is presented.
	
	We consider a scenario, where all UEs and targets are assumed to be randomly distributed in a 300 m$\times$300 m square area with a total of $M=8$ APs, and consider each AP configured with $N=20$ antennas. There are 8 active UEs in the system, $K_{\rm dl}=K_{\rm ul}=4$. Consider the channel to be quasi-static and the path loss exponent of the channel $\alpha=3.7$\cite{Lv2018Uplink}. All UL UEs send the same data power i.e. ${\bigl\{ {{p_{{\rm ul},i}}} \bigr\}_{\forall i \in {\kappa _{\rm ul}}}} = {p_{\rm ul}}=0.1~{\rm W}$. For DL APs, the power is allocated equally based on the number of UEs and the number of targets to be sensed i.e.${\bigl\{ {{p_{\rm dl,i}}} \bigr\}_{\forall i \in {\mathcal {K}_{\rm dl}}}} = {\bigl\{ {{p_{{\rm s},t}}} \bigr\}_{\forall t \in \mathcal {T}}}=0.5~{\rm W}$. The reflection coefficient is set as $\{\alpha_t\}_{\forall t \in {\mathcal {T}_{}}}=0.8$\cite{mao2023communicationsensing}. The signal bandwidth is set as $\Delta f=10$ MHz and the coherent block length is set as $\tau=100$\cite{Van2020Massive}. The AWGN power for DL and UL communication is set to $\sigma _{\rm dl}^2 = \sigma _{\rm ul}^2 = \sigma^2_{\rm s}=-113$ dB\cite{Li2021Network}.
	
	\subsection{Results of System Performance Comparison}
	In this part, we will compare the performance of the NAFD-based ISAC system proposed in this paper with two existing ISAC systems. The specific system design is as follows
	\begin{enumerate}[]
		\item \textit{Time-division duplex (TDD)-based ISAC:} In a TDD-based ISAC system, the coherent time block is divided into three separate modes: sensing, UL communication, and DL communication, which are carried out using the time-division duplex mode\cite {Zhang2022Time}. Symbol lengths for each mode are denoted as ${\tau _s}$, ${\tau _{\rm ul}}$ and ${\tau _{\rm dl}}$, respectively.
		\item \textit{Co-frequency co-time full duplex (CCFD)-based ISAC:} CCFD-based ISAC system employs a single base station to achieve both sensing and simultaneous UL and DL communication\cite {Guo2021Performance}. 
	\end{enumerate}
	
	To ensure a fair comparison, the AP deployment location for TDD-based ISAC is consistent with NAFD-based ISAC and is uniformly distributed on a circular area with a radius of 200 m. Specifically, half of the APs are set to UL mode and the other half are set to DL mode. Conversely, for CCFD-based ISAC, a single base station is deployed at the center of the circular. The self-interference cancellation technique is employed in the CCFD-based ISAC system to achieve a self-interference power $\sigma _{\rm self}^2 =  - 85$ dB\cite{Ma2023Wideband}.
	
	Fig. \ref{fig:cmp_single_muti} presents a comparison and analysis of the sensing performance between single-station independent sensing (CCFD-based ISAC) and multi-static cooperative sensing when only one target is present in the sensing scene. The color scale indicates the LER, with darker colors indicating poorer performance and brighter colors indicating better performance. The results show that single-station independent sensing performs well only near the base station, while multi-static cooperative sensing exhibits better performance in regions near each AP. This highlights the advantage of decentralized AP in achieving larger sensing ranges. Therefore, multi-static cooperative sensing is expected to become a major trend in future developments.
	\begin{figure} 
		\centering
		\subfloat{\includegraphics [scale=0.3]{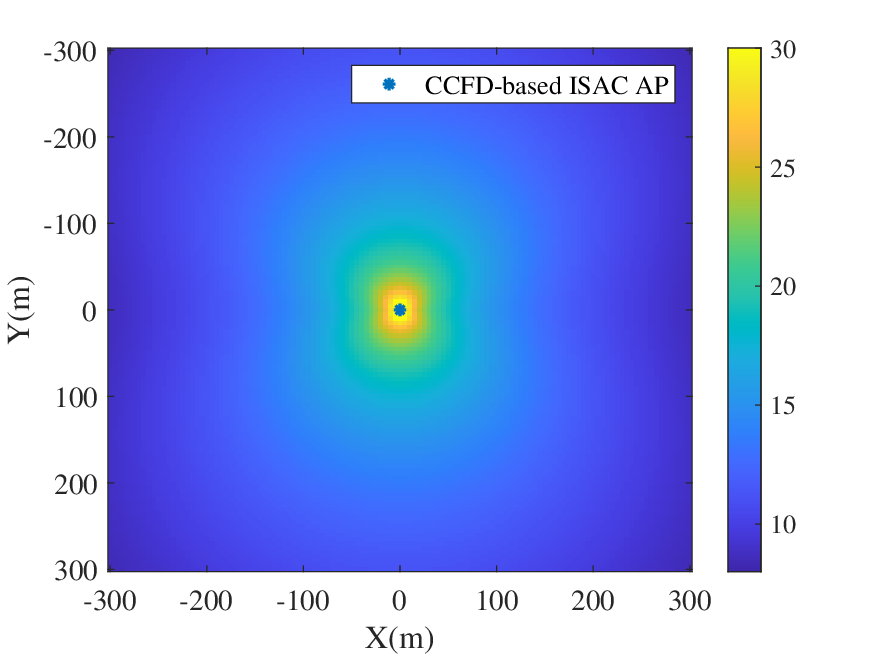}\label {fig:single_sense}}
		\hspace{-2mm}
		\subfloat{\includegraphics [scale=0.3]{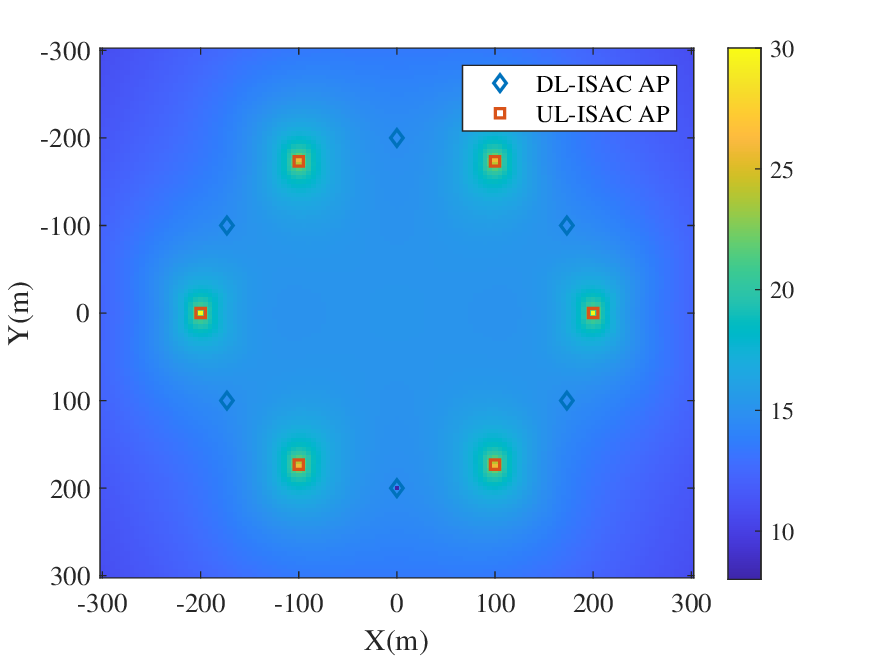}\label {fig:muti_sense}}
		\caption {Comparison of sensing performance between single-station independent sensing and multi-static cooperative sensing . (a) Sensing performance of single-station independent sensing. (b) Sensing performance of multi-static cooperative sensing.}
		\label{fig:cmp_single_muti}
	\end{figure}
	
	Fig. \ref{fig:M} presents the changes in communication and sensing performance with varying numbers of antennas for three different types of ISAC systems. The sensing symbol length of the TDD-based ISAC system is set to $\tau_{\rm s}=20$ and $\tau_{\rm s}=50$. Fig. \ref{fig:rate_M} illustrates the communication performance, where the NAFD-based ISAC system outperforms the CCFD-based and TDD-based ISAC systems. Although the communication performance of the CCFD-based ISAC system improves with the increase in the number of antennas, it still lags behind the NAFD-based ISAC system. The longer the sensing symbol length for the TDD-based ISAC system, the shorter the relative communication symbols, leading to poorer communication performance, whereas the NAFD-based ISAC system is not affected by this issue. Fig. \ref{fig:sense_M} illustrates the sensing performance, where the TDD-based ISAC system's perceptual capability depends only on the length of the sensing symbols, as its sensing is carried out individually. The overall sensing performance of all systems increases with an increase in the number of antennas, as it enhances the signal energy and improves the directivity of the sensed beam. Nevertheless, the sensing performance of the CCFD-based ISAC system is impaired by self-interference. The communication and sensing performance of the NAFD-based ISAC system are mutually interfering, resulting in a degradation of sensing performance that is proportional to the number of antennas. The simulation results demonstrate that the NAFD-based ISAC system with interference cancellation significantly improves communication performance compared to other systems and can outperform them in terms of sensing performance with the appropriate number of antennas. The sensing performance of the NAFD-based ISAC system can also be improved by using appropriate interference management strategies.
	\begin{figure} 
		\centering
		\subfloat{\includegraphics [   scale=0.28]{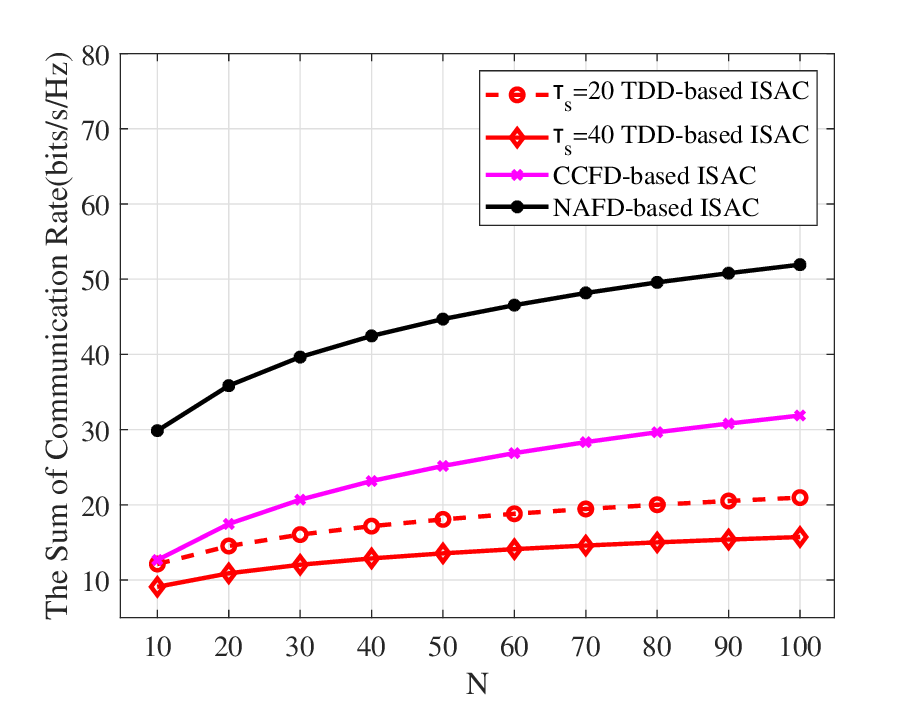} \label {fig:rate_M}}
		\hspace{-2mm}
		\subfloat{\includegraphics [   scale=0.28]{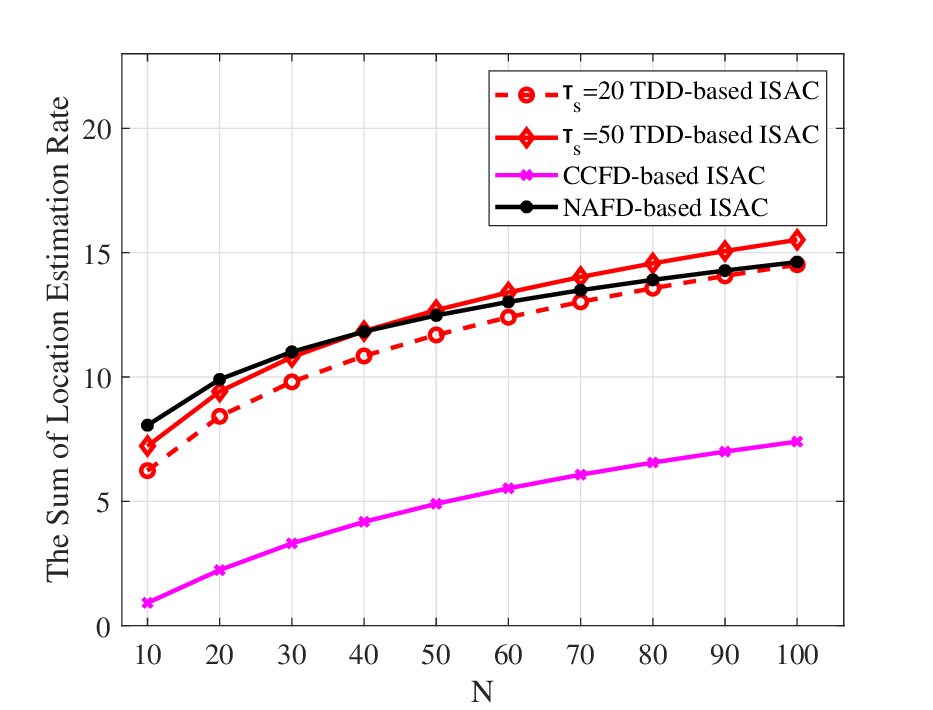}\label {fig:sense_M}}
		\caption {Comparison of system performance with the number of antennas. (a) Comparison of system communication and rate variation with antenna number. (b) Comparison of system LER with antenna number.} 
		\label {fig:M}
	\end{figure}
	
	\subsection{Result of Closed-form Expression Validation}
	In this section, we conduct a validation study of the closed-form expressions for the communication rate and LER through Monte Carlo simulations. Specifically, we analyze the variations of the UL and DL communication rates and position LER obtained from the Monte Carlo simulation with different numbers of antennas. As shown in Fig. \ref {fig:sim_the}, the results of the closed-form expressions are consistent with the Monte Carlo simulation results, and the outcomes are not affected by changes in the number of antennas. These findings confirm the accuracy of the derived closed-form expressions.
	\begin{figure}[htbp]
		\centering
		\subfloat{\includegraphics [scale=0.3]{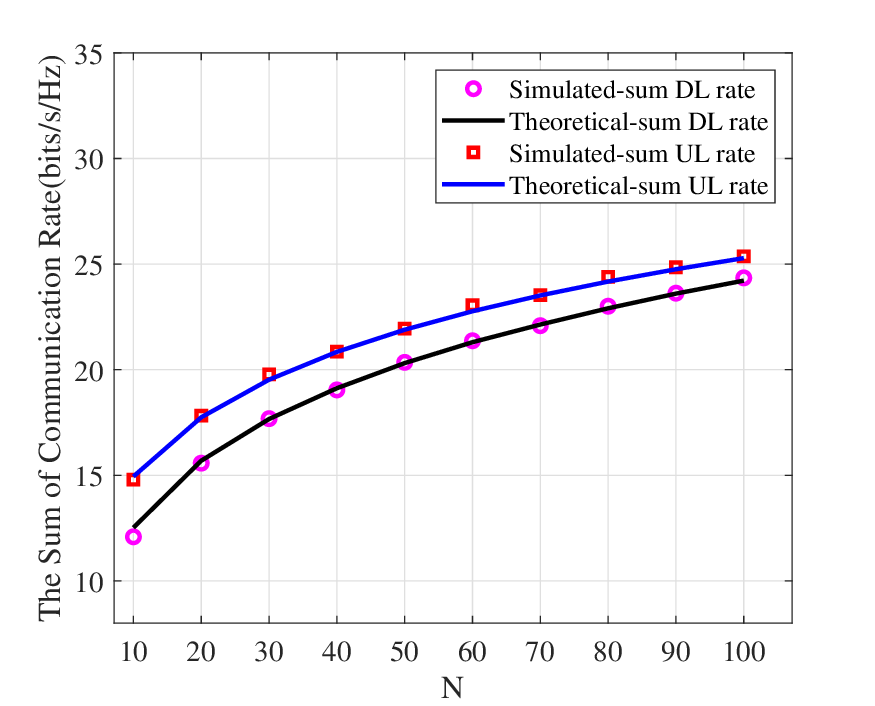}\label {fig:com_sim_the}}
		\hspace{-4mm}
		\subfloat{\includegraphics [scale=0.3]{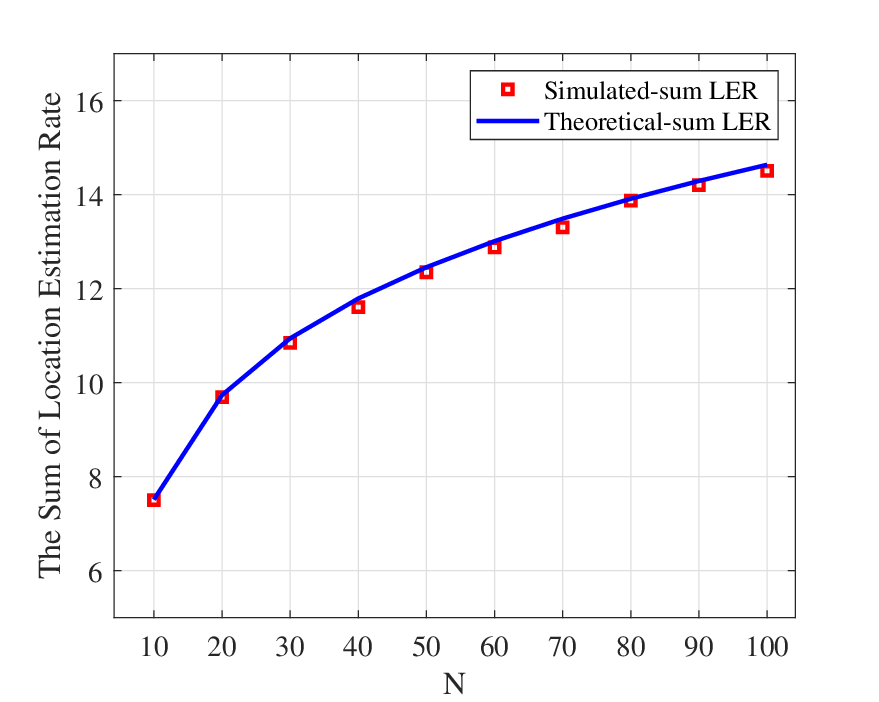}\label {fig:sense_sim_the}}
		\caption {Closed-form expression verification. (a) The closed-form and simulated values of the communication rate vary with the number of antennas. (b) The closed-form and simulated values of the LER vary with the number of antennas.}
		\label{fig:sim_the}
	\end{figure}
	\subsection{Result of duplex mode selection}
	To determine the optimal network structure for DQN, we conduct a study on the effect of different DNN hyperparameters on the convergence performance of the neural network. The reward results of every 20 episodes are averaged and recorded as a convergence curve during the training process. We compare the convergence of neural networks with four network structures using different hyperparameters. Our results show that the single hidden layer consisting of only 10 neurons, i.e. $neure=[10]$, doesn't converge as well as the other network structures. A double hidden layer with 20 neurons, i.e. $neure=[20,20]$, is found to be more effective. Therefore, we use this network structure in subsequent simulations. Furthermore, we study the impact of two key hyperparameters, namely $t_{\rm max}$ and learning rate $lr$ on the convergence performance of the neural network. $t_{\rm max}$ denotes the number of steps during each episode, and a larger one leads to faster convergence but also consumes more training time. On the other hand, the learning rate of the neural network parameters determines the magnitude of the weights along the gradient in the process of updating the network parameters. We find that a smaller learning rate could avoid instability during the training process, resulting in better convergence and avoidance of local optima. However, it may also lead to slower convergence and require more iterations to achieve the desired performance. After simulation comparisons, we select the network with the best convergence performance, i.e. a double hidden layer of 20 neurons, with $t_{max}=10$ and $lr=0.01$ for subsequent simulations.
	\begin{figure}
		\centering
		\includegraphics[scale=0.5]{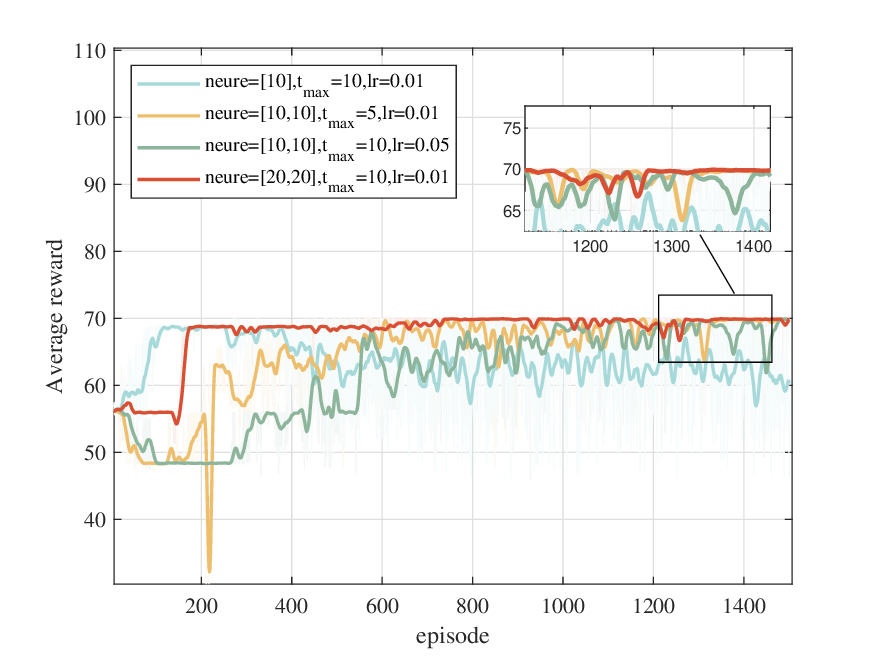}
		\caption {DQN performance under different DNN parameters.}
	\end{figure}
	
	In this section, we introduce four baseline algorithms to compare the performance of the proposed DQN-based ADMO algorithm.
	\begin{enumerate}
		\item \textit{Random Selection (RANDOM):} The RANDOM method is employed to randomly select APs for allocation, which serves as a baseline for system performance without any ADMO.
		\item \textit{Average Selection (AVG):} The AVG method is utilized to assign equal numbers of UL APs and DL APs, which demonstrates system performance when the UL and DL modes of the APs are fixed.
		\item \textit {Exhaustion (EXU):} The EXU method is employed to conduct an exhaustive search, which aims to find the optimal solution by exploring all possible combinations of AP allocation. However, the time complexity of the exhaustive method is typically exponential, which can significantly reduce its efficiency in solving large-scale problems and increase with the size of the problem.
		\item \textit{Q-learning\cite{Zhu2022Load}:} The Q-learning method is a table-based RL approach that utilizes a Q-table to store the Q-value of each state-action pair. In the Q-learning-based ADMO algorithm, the agent, states, actions, and rewards are consistent with those in DQN.
	\end{enumerate}
	
	Fig. \ref{fig:CDF} presents the simulated cumulative probability distribution (CDF) of the objective function values for the proposed algorithm and other baseline algorithms. The objective function is calculated using Eq. (\ref{reward}), while the communication and sensing performance weights in Eq. (\ref{reward}) are set to ${\omega _c}=0.5$ and ${\omega _s}=0.5$ for algorithm comparison. The performance gap between the RANDOM algorithm and the EXU is analyzed as a benchmark. The results show that, at 0.5 CDF, the algorithm with an average allocation of APs outperforms the RANDOM algorithm by about 38\%. This is because the simulation assumes an equal number of UL and DL UEs, resulting in similar communication needs. Therefore, using the average allocation of APs for UL and DL communication can better match UE needs and improve performance. However, the performance gap between the algorithm with the AVG and the EXU algorithm is still significant. Hence, there is room for performance improvement under the fixed AP duplex mode, and dynamic ADMO is necessary. At 0.5 CDF, the Q-learning-based ADMO algorithm improves performance by about 70\% compared to the RANDOM algorithm, and the gap with the EXU algorithm is about 25\%. The result of the DQN algorithm is almost identical to the result of the EXU algorithm, with a performance gap narrowed down to 4\%. Therefore, compared with other algorithms, the DQN-based ADMO algorithm demonstrates superior performance.
	\begin{figure}
		\centering
		\includegraphics[scale=0.5]{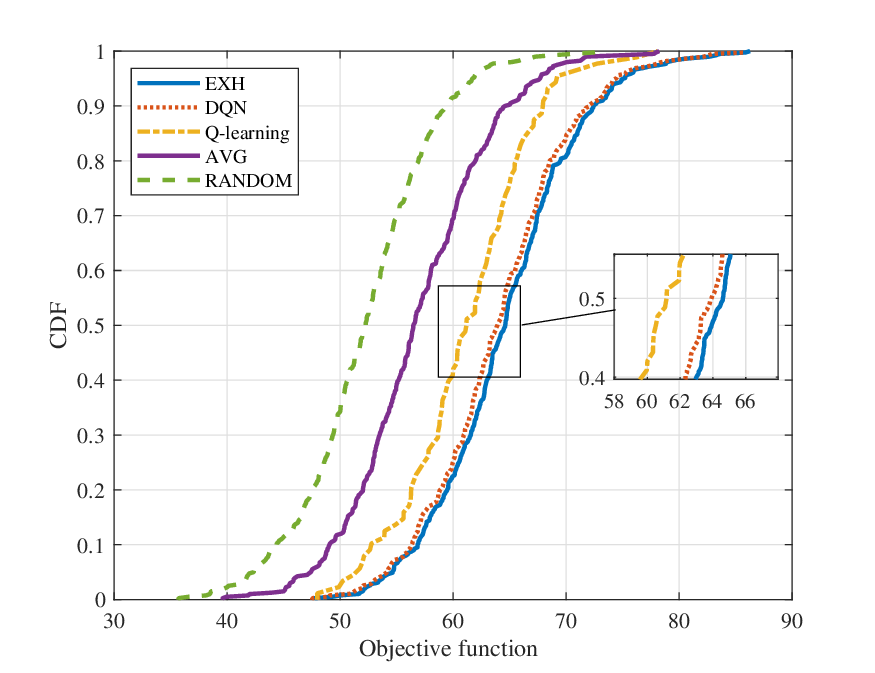}
		\caption {Performance comparison of ADMO algorithms.}
		\label{fig:CDF}
	\end{figure}
	
	Fig. \ref{fig:Pareto} demonstrates the communication and sensing performance corresponding to all the AP duplex modes obtained using the EXU method, which are represented by the blue points in the figure. According to the definition of Pareto solution, a Pareto solution cannot achieve the value of a certain objective function without compromising the values of other objective functions. Therefore, a Pareto solution can be found among the ordinary solutions, as indicated by the black dots in the figure. The Pareto solutions are connected to form a Pareto boundary. Moreover, the red dots in the figure depict the output solutions of DQN when considering different weights of the objectives in Eq. (\ref{reward}), and all the solutions of DQN are either Pareto solutions or approximate Pareto solutions.
	\begin{figure}[htb]
		\centering
		\includegraphics[scale=0.5]{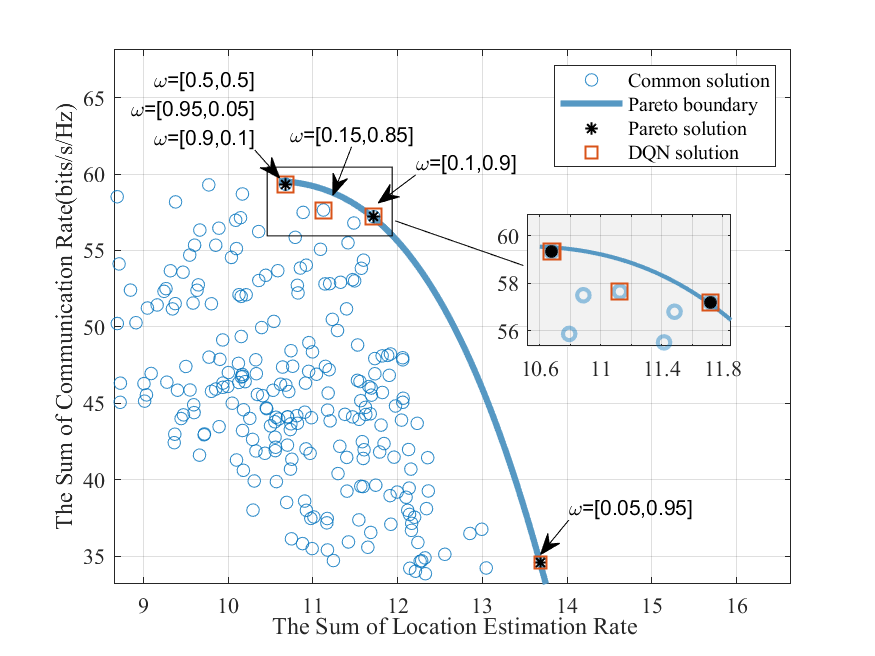}
		\caption {The trade-off between communication and sensing in ADMO.}
		\label{fig:Pareto}
	\end{figure}
	\section{Conclusion}
	In this paper, we propose a NAFD-based ISAC design that supports simultaneous UL and DL communication in muti-static ISAC system. The design employs the NAFD technique, which can partially eliminate CLI. Subsequently, we derive closed-form expressions for the communication rate and LER under imperfect CSI for evaluating the communication performance and sensing performance, respectively. In terms of optimization, we establish a MOOP for communication and sensing ADMO to obtain a trade-off between communication and sensing. Subsequently, we propose a DQN-based ADMO algorithm that sums multiple objective values with different weights as rewards to solve the MOOP. We compare the proposed design with existing TDD-based ISAC and CCFD-based ISAC as benchmark design. The simulation results show that the NAFD-based ISAC system is able to significantly outperform the other ISAC systems in terms of communication performance with essentially no loss of sensing performance. The simulation results verify that the derived closed-form expressions can fit the Monte Carlo simulated values better. Finally, the DQN-based ADMO algorithm can achieve a performance effect close to that of the EXH method with low complexity, and by adopting different target weights, the DQN algorithm can quickly achieve an effect close to the Pareto boundary.
	
	\appendices
	\section{PRELIMINARY RESULTS}
	We first give two lemmas on gamma distribution and projection properties.
	\begin{lemma}
		\label{the1}
		If a random vector $\mathbf {x}$ obeys the distribution $\mathbf {x}\sim \mathcal {C}\mathcal {N}(0,{\sigma ^2}{\mathbf {I}_{\rm N}})$, then $\mathbf{x}$ satisfies ${\mathbf{x}^{\rm H} }\mathbf {x}\sim \Gamma ({\rm N},{\sigma ^2})$. If the random vector $\mathbf{y}$ obeys the distribution $y\sim \Gamma (a,b)$, then $\mathbb {E}\bigl\{\mathbf{y}\bigr\}=ab$ and ${\rm var}\bigl\{\mathbf {y}\bigr\}=ab^2$\cite{Heath2011Multiuser}.
	\end{lemma}
	\begin{lemma}
		\label{the2}
		If a set of independent random vectors $\bigl\{\mathbf {x}_i\bigr\}$ and each vector is distributed as $\Gamma ({k_i},{\theta _i})$, then $\sum\nolimits_i^{} {x_i^{\rm H}{x_i}} \sim \Gamma (k,\theta )$, where $k = \frac{{{{(\sum\nolimits_i^{} {{k_i}{\theta _i}} )}^2}}}{{\sum\nolimits_i^{} {{k_i}\theta _i^2} }}$, $\theta  = \frac{{\sum\nolimits_i^{} {{k_i}\theta _i^2} }}{{\sum\nolimits_i^{} {{k_i}\theta _i^{}} }}$\cite{Zhu2011Performance}.
	\end{lemma}
	\section{Proof of the Theorem 1}
	\label{theorem_1}
	In the aggregation channel, the multipath channel and NLOS paths reflected from the target are statistically independent, and NLOS paths reflected from different targets are also independent. Therefore, the autocorrelation matrix for the channel $\mathbf {h}_{mk}$ between the $k$th UE and AP $m$ is denoted as
	\begin{align}
		\label{form:hmk_autocorrelation}
		{\begin{aligned} &\mathbb{E}\bigl[ {{\mathbf{h}_{mk}}\mathbf{h}_{mk}^{\rm H}} \bigr] \\
				&= \mathbb{E}\bigl[ {{\mathbf{h}_{{{\rm A}_m},{{\rm U}_k}}}\mathbf{h}_{{{\rm A}_m},{{\rm U}_k}}^{\rm H}} \bigr] +\sum\limits_{t \in \mathcal{T}} {\alpha^2_t}\mathbb{E}\bigl[{{{\bigl| {{{\rm h}_{{{\rm T}_t},{{\rm U}_k}}}} \bigr|}^2}}\bigr]\mathbb{E}\bigl[{{\mathbf{h}_{{{\rm A}_m},{{\rm T}_t}}}\mathbf{h}_{{{\rm A}_m},{{\rm T}_t}}^{\rm H}}\bigr].
		\end{aligned}}
	\end{align}
	According to the definition of the channel, we can calculate that $\mathbb{E}\bigl[ {{{\mathbf{h}_{{{\rm A}_m},{{\rm U}_k}}}} {{\mathbf{h}_{{{\rm A}_m},{{\rm U}_k}}^{\rm H}}}} \bigr]=\lambda _{{{\rm A}_m},{{\rm U}_k}}^2{\mathbf{I}_N}$, $\mathbb{E}\bigl[ {{{\bigl\| {{{\rm h}_{{{\rm T}_t},{{\rm U}_k}}}} \bigr\|}^2}} \bigr] ={\bar {\rm d}_{{{\rm T}_t},{{\rm U}_k}}^{ -2{\alpha}}} + \sigma _{{{\rm T}_t},{{\rm U}_k}}^2=\gamma_{{{\rm T}_t},{{\rm U}_k}}$, $\mathbb{E}\bigl[ {{\mathbf{h}_{{{\rm A}_m},{{\rm T}_t}}}\mathbf{h}_{{{\rm A}_m},{{\rm T}_t}}^{\rm H}} \bigr] =( {{{\bar {\rm d}_{{{\rm A}_m},{{\rm T}_t}_{}}^{ - 2{\alpha}}}} + \sigma _{{{\rm A}_m},{{\rm T}_t}}^2})( {{{\bar {\mathbf q}}_{{{\rm A}_m},{{\rm T}_t}}}\bar {\mathbf{q}}_{{{\rm A}_m},{{\rm T}_t}}^{\rm H} + \chi _{{{\rm A}_m},{{\rm T}_t}}^2{{\mathbf{I}}_N}})=\bm{\zeta}_{{{\rm A}_m},{{\rm T}_t}}$, substituting the result into Eq. (\ref{form:hmk_autocorrelation}) gives $\mathbb{E}\bigl[ {{\mathbf{h}_{mk}}\mathbf{h}_{mk}^{\rm H}} \bigr]={{\bm \phi}_{ml}}$. A similar argument can be made to demonstrate that $\mathbb{E}\bigl[ {\mathbf{h}_{{\rm A},mn}^{}{{\mathbf{h}^{\rm H}_{{\rm A},mn}}}} \bigr]= {\bm \phi}^{\rm A}_{mn}$ and $\mathbb{E}\bigl[ {{{\bigl| {{{\rm h}_{{\rm I},u,l}}} \bigr|}^2}} \bigr]={{\rm \phi}^{\rm u}_{ul}}$. This completes the proof.
	
	\section{Proof of theorem 2}
	\label{theorem_2}
	Put expectations into $\gamma^{\rm dl}_{l}$. First, we calculate the molecular part
	\begin{align}
		\label{form:dl}
		\mathbb{E}\big[ {{{\bigl| {{D_l}} \bigr|}^2}} \bigr]={p_{\rm dl,l}}\varepsilon _{{\rm dl},l}\mathbb{E}\big[ {{{\big|{{ \sum\limits_{m \in {\mathcal{M} _{\rm}}}^{} \hat{\mathbf{h}}_{{\rm dl},ml}^{\rm H}{\hat{\mathbf{h}}_{{\rm dl},ml}}} } \big|}^2}} \big].
	\end{align}
	According to Eq. (\ref{form:h_hat_h_hatmk}), $\mathbb{E}\bigl[ {{{\hat {\mathbf {h}}}_{{\rm dl},ml}}\hat {\mathbf {h}}_{{\rm dl},ml}^{\rm H}}\bigr]= {\hat {\mathbf {R}}^{\rm dl}_{ml}}$. Actually, $  {{{\hat {\mathbf{h}}}^{\rm H}_{{\rm dl},ml}}\hat {\mathbf{h}}_{{\rm dl},ml}}=\sum\nolimits_{n = 1}^N {{{{\hat {\mathbf{h}}}^{\rm H}_{{\rm dl},ml,n}}\hat {\mathbf{h}}_{{\rm dl},ml,n}}}$, where ${{\hat {\mathbf {h}}}_{{\rm dl},ml,n}}$ denotes the estimated channel from the $n$th antenna of the AP $m$ to the UE $l$. It is known that ${{\hat {\mathbf {h}}}_{{\rm dl},ml}} \sim \mathcal {CN}(0,\hat {\mathbf {R}}^{\rm dl}_{ml})$, so ${{\hat {\mathbf{h}}}_{{\rm dl},ml,n}}\sim \mathcal{CN}\bigl(0,\bigl[\hat{\mathbf{R}}^{\rm dl}_{ml}\bigr]_n\bigr)$. According to the Lemma \ref{the1}, ${\hat {\mathbf{h}}}^{\rm H}_{{\rm dl},ml,n}{{\hat {\mathbf{h}}}_{{\rm dl},ml,n}}\sim \bigl( {1,{{\bigl[ {{\hat{\mathbf{R}}^{\rm dl}_{ml}}} \bigr]}_n}} \bigr)$. Combined with the Lemma \ref {the2}, we can get $\hat{\mathbf{h}}^{\rm H}_{{\rm dl},ml}{\hat {\mathbf{h}}_{{\rm dl},ml}}\sim \Gamma \bigl( {{{\hat k}_{{\rm dl},ml}},{{\hat \theta }_{{\rm dl},ml}}} \bigr)$ obeying gamma distribution, where 
	\begin{align}
		{{\hat k}_{{\rm dl},ml}}&={\frac{\bigl({\sum\nolimits_{n = 1}^N\bigl[ {{\hat{\mathbf{R}}^{\rm dl}_{ml}}} \bigr]_{n}}\bigr)^2}{{\sum\nolimits_{n = 1}^N {\bigl[ {{\hat{\mathbf{R}}^{\rm dl}_{ml}}} \bigr]_n^2} }}},\\ 
		{{\hat \theta }_{{\rm dl},ml}}&=\frac{{\sum\nolimits_{n = 1}^N {\bigl[ {{\hat{\mathbf{R}}^{\rm dl}_{ml}}} \bigr]_n^2} }}{{\sum\nolimits_{n = 1}^N {\bigl[ {{\hat{\mathbf{R}}^{\rm dl}_{ml}}} \bigr]_n^{}}}}. 
	\end{align}
	As previously stated in Lemma 2, the distrbution of ${\sum\nolimits_{m \in {\mathcal{M} _{\rm dl}}}^{} {\hat {\mathbf{h}}_{{\rm dl},ml}^{{\rm H}}\hat {\mathbf{h}}_{{\rm dl},ml}} }$ can be derived. Combining Lemma 1, we can get 
	\begin{align}
		\label{form:sum_h}
		&\mathbb{E}\bigg[ {{{\bigg| {\sum\limits_{m \in {\mathcal{M} _{\rm}}}^{} { \hat{\mathbf{h}}_{{\rm dl},ml}^{\rm H}{\hat{\mathbf{h}}_{{\rm dl},ml}}} } \bigg|}^2}} \bigg]\nonumber \\
		&=\sum\limits_{m \in {\mathcal{M} _{\rm}}}^{} {{\hat k}_{{\rm dl},ml}}\hat \theta _{{\rm dl},ml}^2 + {{(\sum\limits_{m \in {\mathcal{M} _{\rm}}}^{} {{{\hat k}_{{\rm dl},ml}}\hat \theta _{{\rm dl},ml}^{}})^2}}.
	\end{align}
	By substituting the aforementioned equation into Eq. (\ref{form:dl}), we can derive the numerator term.
	Next, the procedure involves the identification of individual components in the denominator. Since the channels are orthogonal between different UEs, the DL interference from other UEs can be expressed as 
	\begin{align}
		\mathbb{E}\bigl[ {{{\bigl| {\mathcal{N}_l^{\rm dl}} \bigr|}^2}} \bigr]=\sum\limits_{l' \in {\mathcal{K} _{\rm dl}}\backslash \{ l\} }^{}\sum\limits_{m \in {\mathcal{M}_{\rm}}}^{}{{p_{{\rm dl},l'}\varepsilon^2_{{\rm dl},l'}}}\mathbb{E}\big[{ {\big| \hat {\mathbf{h}}_{{\rm dl},ml}^{\rm H}\hat {\mathbf{h}}_{{\rm dl},ml'}\big|^2}}\big], 
	\end{align}
	where $\mathbb{E}\bigl[ {{{\bigl| \hat {\mathbf{h}}_{{\rm dl},ml}^{\rm H}\hat {\mathbf{h}}_{{\rm dl},ml'}\bigr|}^2}} \bigr]=\mathbb{E}\bigl[ {\hat {\mathbf{h}}_{{\rm dl},ml}^{\rm H}{{\hat {\mathbf{R}}}^{{\rm dl}}_{m{l'}}}\hat {\mathbf{h}}_{{\rm dl},ml}^{\rm H}} \bigr]=\mathbb{E}\bigl[\sum\nolimits_{n = 1}^N {{\hat {\mathbf{h}}_{{\rm dl},ml,n}^{\rm H}}{{\bigl[ {{{\hat {\mathbf{R}}}^{\rm dl}_{m{l'}}}} \bigr]}_n}{\hat {\mathbf{h}}_{{\rm dl},ml,n}}}\bigr]$. Using an approach similar to Eq. (\ref{form:sum_h}) yields, we can deduce that $\mathbb{E}\bigl[ {\hat {\mathbf{h}}_{{\rm dl},ml}^{\rm H}{{\hat {\mathbf{R}}}^{\rm dl}_{m{l'}}}\hat {\mathbf{h}}_{{\rm dl},ml}^{\rm H}} \bigr]={{{\rm tr}\bigl({{{{\hat {\mathbf{R}}}^{\rm dl}_{ml}}}}{{{{\hat {\mathbf{R}}}^{\rm dl}_{m{l'}}}}}\bigr)}}$. So the first term of the denominator is
	\begin{align}
		\label{form:I_inter_dl_res}
		\mathbb{E}\bigl[ {{{\bigl| {\mathcal{N}_l^{\rm dl}} \bigr|}^2}} \bigr]&=\sum\limits_{l' \in {\mathcal{K} _{\rm dl}}\backslash \{ l\} }^{}\sum\limits_{m \in {\mathcal{M}  _{\rm}}}^{} {{p_{{\rm dl},l'}\varepsilon^2_{{\rm dl},l'}}}{{{\rm tr}\bigl({{{{\hat {\mathbf{R}}}^{\rm dl}_{ml}}}}{{{{\hat {\mathbf{R}}}^{\rm dl}_{m{l'}}}}}\bigr)}}.
	\end{align}
	The second term in the denominator is similarly obtained through a similar approach, which can be expressed as 
	\begin{align}
		\mathbb{E}\bigl[\bigl|\mathcal{N}_l^{\rm error} \bigr|^2\bigr]= \sum\limits_{k \in {\mathcal{K} _{dl}}}\sum\limits_{m \in {\mathcal{M} _{\rm dl}}} {p_{{\rm dl},k}\varepsilon^2_{{\rm dl},k}}{\rm tr}\bigl(\hat {\mathbf{R}}^{\rm dl}_{mk} {\bm{\theta}}^{\rm dl}_{ml}\bigr).
	\end{align}
	It can be demonstrated that the DL channel estimation error and sensing beamforming are independent for different APs. The third term in the denominator can be expressed as 
	\begin{align}
		\label{form:dl_sense}
		\mathbb{E}\bigl[\bigl| \mathcal{N}_l^{\rm CLI-s} \bigr|^2\bigr]&=\sum\limits_{m \in {\mathcal{M} _{\rm}}} \mathbb{E}\bigl[\bigl| {\mathbf{e}_{{\rm dl},ml}^{\rm H}{ {{\mathbf{q}}_{{{ \rm A}_m},{{ \rm T}_t}}}} }\bigr|^2\bigr].
	\end{align}
	Following the definition of the steering vector in Eq. (\ref{form3:qat}), we can derive that 
	\begin{align}
		\label{form:qq_H}
		\mathbb{E}\bigl[{ {{\rm \mathbf{q}}_{{{ \rm A}_m},{{ \rm T}_t}}}}{ {{\mathbf{q}}^{\rm H}_{{{ \rm A}_m},{{ \rm T}_t}}}}\bigr]&={ {\bar  {\mathbf{q}}_{{{ \rm A}_m},{{ \rm T}_t}}}}{ {\bar {\rm \mathbf{q}}^{\rm H}_{{{ \rm A}_m},{{ \rm T}_t}}}}+\chi _{{{ \rm A}_m},{{ \rm T}_t}}^2{\mathbf{I}_N} ={\bm{\psi} _{mt}}.
	\end{align} 
	Subsequently, Eq. (\ref{form:dl_sense}) can be expressed as $\mathbb{E}\bigl[\sum\nolimits_{n = 1}^N\mathbf{e}_{{\rm dl},ml,n}^{\rm H}\bigl[{\bm{\psi} _{mt}}\bigr]_n\mathbf{e}_{{\rm dl},ml,n}\bigr]$. Using a derivation process similar to the above, it can be obtained
	\begin{align}
		\label{form:I_sense_dl_res}
		\mathbb{E}\bigl[\bigl| \mathcal{N}_l^{\rm CLI-s} \bigr|^2\bigr]&=\sum\limits_{t \in \mathcal{T}}\sum\limits_{m \in {\mathcal{M} _{\rm}}} \frac{{{p_{{\rm s},t}}}}{N}{\rm tr}\bigl({\bm{\psi} _{mt}}{\bm{\theta}}_{{\rm dl},ml}\bigr).
	\end{align}
	The final term is the CLI term, which, due to the orthogonality of the channels utilized by different UL UEs, can be expressed as
	\begin{align}
		\label{form:I_CLI_dl}
		\mathbb{E}\bigl[\bigl| \mathcal{N}_l^{\rm CLI-c} \bigr|^2\bigr]&=\sum\limits_{u \in {\mathcal{K} _{\rm ul}}}{{p_{{\rm ul},u}}} \mathbb{E}\bigl[\bigl|{{h_{{\rm I},u,l}}}\bigr|^2\bigr]=\sum\limits_{u \in {\mathcal{K} _{\rm ul}}}{{p_{{\rm ul},u}}}{{\rm \phi}^{\rm u}_{ul}}.
	\end{align}
	Combining all results completes the proof.
	\section{Proof of theorem 3}
	\label{theorem_3}
	Put expectations into $\gamma^{\rm ul}_{u}$. The derivations of ${\mathbb{E}\bigl[\bigl|{\mathcal{D}_u} \bigr|^2\bigr]}$, ${\mathbb{E}\bigl[\bigl|\mathcal{N}_u^{\rm ul}  \bigr|^2\bigr]}$, ${\mathbb{E}\bigl[\bigl|\mathcal{N}_u^{\rm error} \bigr|^2\bigr]}$ and ${\mathbb{E}\bigl[\bigl|\mathcal{N}_u^{\rm CLI-s}  \bigr|^2\bigr]}$ in Eq. (\ref{form:sinr_ul_res}) align with those in the Appendix \ref{theorem_2}, and will not be repeated here. The third term of the denominator is expressed as
	\begin{align}
		\label{form:I_CLI_com_ul}
		{\mathbb{E}\bigl[\bigl|\mathcal{N}_u^{\rm CLI-c}  \bigr|^2\bigr]}
		=\sum\limits_{m \in {\mathcal{M} _{\rm}}}^{}\sum\limits_{n \in {\mathcal{M} _{\rm }}}^{} {{{p_{{\rm dl},j}}} }{\mathbb{E}\bigg[\bigg|\sum\limits_{j \in {\mathcal{K} _{\rm dl}}}^{}\mathbf{v}_{nu}^{\rm H}{\mathbf{e}}_{{\rm A},mn}\mathbf{w}_{m,j}^{\rm c}\bigg|^2\bigg]},
	\end{align}
	Since ${\mathbf {e}}_{{\rm A},mn}$ and $\mathbf {v}_{nu}^{\rm H}$, $\mathbf {w}_{m,j}^{\rm c}$ are not correlated with each other, each of the terms in the above equation can be calculated individually. ${\mathbb{E}\bigl[\bigl\|{\mathbf{e}}_{{\rm A},mn}\bigr\|^2\bigr]}={\mathbb{E}\bigl[\bigl|{\mathbf{e}}_{{\rm A},mn}^{\rm H}{\mathbf{e}}_{{\rm A},mn}\bigr|\bigr]}={\rm tr}\bigl(\bm{\theta}_{mn}^{\rm A}\bigr)$. Combining the Lemma \ref {the1} and Lemma \ref {the2}, we can get
	\begin{align}
		\label{form:E2_CLI_com_ul}
		{\mathbb{E}\bigg[\bigg|\sum\limits_{j \in {\mathcal{K} _{\rm dl}}}^{}\mathbf{v}_{nu}^{\rm H}\mathbf{w}_{m,j}^{\rm c}\bigg|^2\bigg]}
		&=\sum\limits_{j \in {\mathcal{K} _{\rm dl}}}^{}\varepsilon _{{\rm ul },u}^{2}\varepsilon _{{\rm dl},j}^{2}{\mathbb{E}\bigl[\bigl| \hat{\mathbf{h}}_{{\rm ul},mu}^{\rm H}\hat{\mathbf{h}}_{{\rm dl},m,j}\bigr|^2\bigr]} \nonumber\\
		&=\sum\limits_{j \in {\mathcal{K} _{\rm dl}}}^{}\varepsilon _{{\rm ul},u}^{2}\varepsilon _{{\rm dl},j}^{2}{\rm tr}\bigl(\hat {\mathbf{R}}^{\rm ul}_{nu}\hat {\mathbf{R}}^{\rm dl}_{mj}\bigr).
	\end{align}
	By substituting the aforementioned equation into Eq. (\ref{form:I_CLI_com_ul}), we can derive the third term of the denominator. The last term in the denominator can be expressed as
	\begin{align}
		\label{form:n_ul}
		\mathbb{E}\bigl[\bigl|\mathcal{N}_u^{\rm noise}\bigr|^2\bigr]&=\sigma^2_{\rm ul}\sum\limits_{n \in {\mathcal{M} _{\rm }}} \mathbb{E}\bigl[\bigl\| {\mathbf{v}_{nu}^{\rm H}}\bigr\|^2\bigr]=\sigma^2_{\rm ul}\sum\limits_{n \in {\mathcal{M} _{\rm }}} {\rm tr}(\hat{\mathbf{R}}^{\rm ul}_{nu}).
	\end{align}
	Combining all results completes the proof.
	\section{Proof of theorem 4}
	\label{theorem_4}
	The received signals are separately mapped onto distinct sensing symbols, enabling accurate identification of the corresponding signals $y_{m,n,t}$. Concerning the active MIMO radar signal model, the detailed modeling of ${\mathbf{y}_{m,n,t}}$ is represented as
	\begin{equation}
		\label{form:sense_h_nlos_detail}
		{\begin{aligned}{\mathbf{y}_{m,n,t}} =\sum\limits_{i \in \mathcal{T}}{{\sqrt {{p_{{\rm s},t}}} }{\alpha _i}\eta _{m,n,i}\mathbf{a}_{mi}\otimes\mathbf{b}_{ni}\mathbf{w}_{mt}^{\rm s}}e^{-j2\pi\frac{\Delta f}{c}d_{m,n,i}}+{\mathbf{z}},
		\end{aligned}}
	\end{equation}
	where $\mathbf{a}_{mi}={\mathbf{q}}_{{{ \rm A}_m},{{ \rm T}_t}}$, $\mathbf{b}_{in}={\mathbf{q}}_{{{ \rm T}_i},{{ \rm A}_n}}$, $\eta _{m,n,i}={\lambda _{{{ \rm A}_m},{{ \rm T}_i}}}{\lambda _{{{ \rm T}_i},{{ \rm A}_n}}}$ represents large-scale fading. 
	First, we derive the closed-form expression for the noise $\mathbf{z}$. The terms of Eq. (\ref{form:noise}) are obtained through a method analogous to that used in Appendix \ref{theorem_2}. These terms can be expressed as
	\begin{align}
		{\mathbb{E}\bigl[\bigl\|{\bar {\mathbf { y}}}_{{\rm dl},n}\bigr\|^2\bigr]}
		&=\sum\limits_{m \in {\mathcal{M} _{\rm }}}^{}\sum\limits_{j \in {\mathcal{K} _{\rm dl}}}^{}{{{p_{{\rm dl},j}}} }\varepsilon _{{\rm dl},j}^{2}{\rm tr}\bigl(\bm {\theta}_{mn}^{\rm A}\hat {\mathbf{R}}^{\rm dl}_{mj}\bigr)=\sigma^2_{{\rm dl},n},\\
		\mathbb{E}\bigl[\bigl\|{\bar {\mathbf { y}}}_{{\rm ul},n}\bigr\|^2\bigr]&=\sum\limits_{k \in {\mathcal{K} _{\rm ul}}} {{p_{{\rm ul},k}}}{\rm tr}\bigl(\hat{\mathbf{R}}^{\rm ul}_{nk}\bigr)=\sigma^2_{{\rm ul},n}.
	\end{align}
	Based on the above equation the CRLB of the estimated parameters can be obtained. The CRLB is obtained by computing the inverse of the Bayesian Fisher Information Matrix (FIM), represented as
	\begin{align}
		\label{form:crlb_fim}
		{\begin{aligned}\mathbb{E}\bigl[ {({{\hat {\bm{\gamma}} }_{m,n,t}} - {\bm {\gamma} _{m,n,t}}){{({{\hat {\bm{\gamma}}  }_{m,n,t}} - {\bm {\gamma} _{m,n,t}})}^{\rm T}}} \bigr] \succeq \mathbf{J}_{m,n,t}^{ - 1},
		\end{aligned}}
	\end{align}
	where ${\bm {\gamma} _{m,n,t}} = {\bigl [ {{d_{m,n,t}},{\theta _{mt}},{\phi _{tn}}} \bigr]{\rm ^T}}$ is the true value of the estimated parameter of target $t$, $d_{m,n,t}=d_{{\rm A}_{m},{\rm T}_{t}}+d_{{\rm T}_{t},{\rm A}_{n}}$, $\hat {\bm {\gamma}}  _{m,n,t}= {\bigl [ {{\hat d_{m,n,t}},{\hat \theta _{mt}},{\hat \phi _{tn}}} \bigr]{\rm ^T}}$  is the estimated value of the parameter of target $t$ based on Eq. (\ref{form:sense_h_nlos_detail}), $\mathbf {J}_{m,n,t}^{} \in {{\mathbb R}^{3 \times 3}}$ is the FIM for the unknown matrix $\bm {\gamma} _{m,n,t}$. The FIM can be computed through
	\begin{align}
		\label{form:fim}
		{\begin{aligned}\mathbf{J}_{m,n,t}=\frac{1}{\sigma^2_{z}}(\frac{{{\rm d}{\mathbf{y}_{m,n,t}}}}{{{\rm d}\bm {\gamma} _{m,n,t}^T}})^{\rm H}(\frac{{{\rm d}{\mathbf{y}_{m,n,t}}}}{{{\rm d}\bm {\gamma} _{m,n,t}^T}}),
		\end{aligned}}
	\end{align}
	where ${\sigma^2_{z}}=\sigma^2_{{\rm dl},n}+\sigma^2_{{\rm ul},n}+\sigma^2_{\rm s}$, $\frac {{{\rm d}{\mathbf {y}_{m,n,t}}}}{{{\rm d}\bm {\gamma} _{m,n,t}^T}}$  expressed as
	\begin{align}
		{\begin{aligned}
				\frac{{{\rm d}{\mathbf{y}_{m,n,t}}}}{{{\rm d}\bm {\gamma}_{m,n,t}^T}}=\zeta_{m,n,t} \times\bigl[& -j2\pi \frac{\Delta f}{c}\mathbf{a}_{mt}\otimes\mathbf{b}_{tn},\\
				&\frac{{\rm d}{\mathbf{a}_{mt}}}{{\rm d}\theta_{mt}}\otimes{\mathbf{b}_{tn}},\mathbf{a}_{mt}\otimes\frac{{\rm d}\mathbf{b}_{tn}}{{\rm d}\phi_{tn}}\bigr],
		\end{aligned}}
	\end{align}
	where $\zeta_{m,n,t}={\alpha^2_t}{{{\bigl\| {\mathbf{w}_{m,t}^{\rm s}} \bigr\|}^2}}e^{-j2\pi \frac{\Delta f}{c}d_{m,n,t}}$, and
	\begin{subequations}
		\begin{align}
			&\frac{{\rm d}{\mathbf{a}_{mt}}}{{\rm d}\theta_{mt}}=-j\frac{2\pi}{\lambda}\bigl[\mathbf{p}^{\rm T}_{m1}\mathbf{k}_{mt}\mathbf{a}_{mt,1},...,\mathbf{p}^{\rm T}_{mN}\mathbf{k}_{mt}\mathbf{a}_{mt,N}\bigr], \\
			&\frac{{\rm d}{\mathbf{b}_{tn}}}{{\rm d}\phi_{tn}}=-j\frac{2\pi}{\lambda}\bigl[\mathbf{p}^{\rm T}_{n1}\mathbf{k}_{tn}\mathbf{b}_{tn,1},...,\mathbf{p}^{\rm T}_{nN}\mathbf{k}_{tn}\mathbf{b}_{tn,N}\bigr],
		\end{align}
	\end{subequations}
	where $\mathbf{k}_{mt}=[-\sin(\theta_{mt}),\cos(\theta_{mt})]^{\rm T}$, and  $\mathbf{k}_{nt}=[-\sin(\phi_{tn}),\cos(\phi_{tn})]^{\rm T}$.
	In the context of large antenna arrays, the orthogonal nature of multiple paths from distinct directions allows for the asymptotic approximation of the non-diagonal sub-matrix of the FIM as the zero matrix. This results in an FIM approximation matrix of a block diagonal structure. The block diagonal FIM matrix enables the calculation of the closed-form expression for the CRLB of each estimated parameter. Inverting Eq. (\ref{form:fim}) and taking out the block for $d_{m,n,t}$, $\theta_{mt}$ and $\phi_{tn}$ yields Theorem \ref{the:LER}. 
	
	This completes the proof.

  \bibliographystyle{IEEEtran}
	\bibliography{ref} 

\begin{thebibliography}{10}
\providecommand{\url}[1]{#1}
\csname url@samestyle\endcsname
\providecommand{\newblock}{\relax}
\providecommand{\bibinfo}[2]{#2}
\providecommand{\BIBentrySTDinterwordspacing}{\spaceskip=0pt\relax}
\providecommand{\BIBentryALTinterwordstretchfactor}{4}
\providecommand{\BIBentryALTinterwordspacing}{\spaceskip=\fontdimen2\font plus
\BIBentryALTinterwordstretchfactor\fontdimen3\font minus
  \fontdimen4\font\relax}
\providecommand{\BIBforeignlanguage}[2]{{%
\expandafter\ifx\csname l@#1\endcsname\relax
\typeout{** WARNING: IEEEtran.bst: No hyphenation pattern has been}%
\typeout{** loaded for the language `#1'. Using the pattern for}%
\typeout{** the default language instead.}%
\else
\language=\csname l@#1\endcsname
\fi
#2}}
\providecommand{\BIBdecl}{\relax}
\BIBdecl

\bibitem{Kaushik2024Toward}
A.~Kaushik, R.~Singh, S.~Dayarathna, and Senanayake, ``{Toward Integrated
  Sensing and Communications for 6G: Key Enabling Technologies,
  Standardization, and Challenges},'' \emph{IEEE Commun. Stand. Mag.}, vol.~8,
  no.~2, pp. 52--59, Jun. 2024.

\bibitem{banafaa20236g}
M.~Banafaa, I.~Shayea, J.~Din, M.~H. Azmi, A.~Alashbi, Y.~I. Daradkeh, and
  A.~Alhammadi, ``{6G Mobile Communication Technology: Requirements, Targets,
  Applications, Challenges, advantages, and opportunities},'' \emph{Alexandria
  Eng. J.}, vol.~64, pp. 245--274, Feb. 2023.

\bibitem{Liu2018Waveform}
F.~Liu, L.~Zhou, C.~Masouros, A.~Li, W.~Luo, and A.~Petropulu, ``{Toward
  Dual-functional Radar-Communication Systems: Optimal Waveform Design},''
  \emph{IEEE Trans. Signal Process.}, vol.~66, no.~16, pp. 4264--4279, Aug.
  2018.

\bibitem{Sun2022Optimal}
P.~Sun, H.~Dai, and B.~Wang, ``{Optimal transmit Beamforming for Near-Field
  Integrated Sensing and Wireless Power Transfer Systems},'' \emph{Intell.
  Converged Networks}, vol.~3, no.~4, pp. 378--386, Dec. 2022.

\bibitem{Guo2023Joint}
T.~Guo, X.~Li, M.~Mei, Z.~Yang, J.~Shi, K.-K. Wong, and Z.~Zhang, ``{Joint
  Communication and Sensing Design in Coal Mine Safety Monitoring: 3-D Phase
  Beamforming for RIS-Assisted Wireless Networks},'' \emph{IEEE Internet Things
  J.}, vol.~10, no.~13, pp. 11\,306--11\,315, Jul. 2023.

\bibitem{Rahman2020Framework}
M.~L. Rahman, J.~A. Zhang, X.~Huang, Y.~J. Guo, and R.~W. Heath, ``{Framework
  for a Perceptive Mobile Network Using Joint Communication and Radar
  Sensing},'' \emph{IEEE Trans. Aerosp. Electron. Syst.}, vol.~56, no.~3, pp.
  1926--1941, Jun. 2020.

\bibitem{Huang2022Coordinated}
Y.~Huang, Y.~Fang, X.~Li, and J.~Xu, ``{Coordinated power control for network
  integrated sensing and communication},'' \emph{IEEE Trans. Veh. Technol.},
  vol.~71, no.~12, pp. 13\,361--13\,365, Dec. 2022.

\bibitem{Fang2021Cell}
Y.~Fang, L.~Qiu, X.~Liang, and C.~Ren, ``{Cell-Free Massive MIMO Systems With
  Oscillator Phase Noise: Performance Analysis and Power Control},'' \emph{IEEE
  Trans. Veh. Technol.}, vol.~70, no.~10, pp. 10\,048--10\,064, Oct. 2021.

\bibitem{Ammar2022User}
H.~A. Ammar, R.~Adve, S.~Shahbazpanahi, G.~Boudreau, and K.~V. Srinivas,
  ``{User-Centric Cell-Free Massive MIMO Networks: A Survey of Opportunities,
  Challenges and Solutions},'' \emph{IEEE Communications Surveys \& Tutorials},
  vol.~24, no.~1, pp. 611--652, Firstquarter 2022.

\bibitem{Sakhnini2022Target}
A.~Sakhnini, M.~Guenach, A.~Bourdoux, H.~Sahli, and S.~Pollin, ``{A Target
  Detection Analysis in Cell-Free Massive MIMO Joint Communication and Radar
  Systems},'' in \emph{IEEE Int. Conf. Commun. (ICC)}, May 2022, pp.
  2567--2572.

\bibitem{Behdad2022PowerAF}
Z.~Behdad, {\"O}.~T. Demir, K.~W. Sung, E.~Bj{\"o}rnson, and C.~Cavdar,
  ``{Power Allocation for Joint Communication and Sensing in Cell-Free Massive
  MIMO},'' in \emph{Proc. IEEE Glob. Commun. Conf. (GLOBECOM)}, Dec. 2022, pp.
  4081--4086.

\bibitem{demirhan2024cellfree}
U.~Demirhan and A.~Alkhateeb, ``{Cell-Free Joint Sensing and Communication
  MIMO: A Max-Min Fair Beamforming Approach},'' in \emph{Conf. Rec. Asilomar
  Conf. Signals Syst. Comput. (ACSSC)}, Oct. 2023, pp. 381--386.

\bibitem{mao2023communicationsensing}
W.~Mao, Y.~Lu, C.~Chi, B.~Ai, Z.~Zhong, and Z.~Ding, ``{Communication-Sensing
  Region for Cell-Free Massive MIMO ISAC Systems},'' \emph{IEEE Trans. Wireless
  Commun.}, Apr. 2024.

\bibitem{Wei2020Load}
S.~Wei, T.~Li, and W.~Wu, ``{Load Optimization of Joint User Association and
  Dynamic TDD in Ultra Dense Networks},'' in \emph{Int. Wirel. Commun. Mob.
  Comput. (IWCMC)}, Jun. 2020, pp. 545--550.

\bibitem{Soret2019Queueing}
B.~Soret, P.~Popovski, and K.~Stern, ``{A Queueing Approach to the Latency of
  Decoupled UL/DL With Flexible TDD and Asymmetric Services},'' \emph{IEEE
  Wireless Commun. Lett.}, vol.~8, no.~6, pp. 1704--1708, Dec. 2019.

\bibitem{Sit2011Extension}
Y.~L. Sit, L.~Reichardt, C.~Sturm, and T.~Zwick, ``{Extension of the OFDM Joint
  Radar-Communication System for a Multipath, Multiuser Scenario},'' in
  \emph{2011 IEEE RadarCon (RADAR)}, May 2011, pp. 718--723.

\bibitem{Wang2019Performance}
D.~Wang, M.~Wang, P.~Zhu, J.~Li, J.~Wang, and X.~You, ``{Performance of
  Network-Assisted Full-Duplex for Cell-Free Massive MIMO},'' \emph{IEEE Trans.
  Commun.}, vol.~68, no.~3, pp. 1464--1478, Mar. 2020.

\bibitem{Li2021Network}
J.~Li, Q.~Lv, P.~Zhu, D.~Wang, J.~Wang, and X.~You, ``{Network-Assisted
  Full-Duplex Distributed Massive MIMO Systems With Beamforming Training Based
  CSI Estimation},'' \emph{IEEE Trans. Wireless Commun.}, vol.~20, no.~4, pp.
  2190--2204, Apr. 2021.

\bibitem{Liu2023Performance}
M.~Liu, M.~Yang, H.~Li, K.~Zeng, Z.~Zhang, A.~Nallanathan, G.~Wang, and
  L.~Hanzo, ``{Performance Analysis and Power Allocation for Cooperative ISAC
  Networks},'' \emph{IEEE Internet Things J.}, vol.~10, no.~7, pp. 6336--6351,
  Apr. 2023.

\bibitem{behdad2024multistatic}
Z.~Behdad, {\"O}.~T. Demir, K.~W. Sung, E.~Bj{\"o}rnson, and C.~Cavdar,
  ``{Multi-Static Target Detection and Power Allocation for Integrated Sensing
  and Communication in Cell-Free Massive MIMO},'' \emph{IEEE Trans. Wireless
  Commun.}, 2024.

\bibitem{Shi2022Device}
Q.~Shi, L.~Liu, S.~Zhang, and S.~Cui, ``{Device-free Sensing in OFDM Cellular
  Network},'' \emph{IEEE J. Sel. Areas Commun.}, vol.~40, no.~6, pp.
  1838--1853, Jun. 2022.

\bibitem{Zhu2021Optimization}
Y.~Zhu, J.~Li, P.~Zhu, H.~Wu, D.~Wang, and X.~You, ``{Optimization of Duplex
  Mode Selection for Network-Assisted Full-Duplex Cell-Free Massive MIMO
  Systems},'' \emph{IEEE Commun. Lett.}, vol.~25, no.~11, pp. 3649--3653, Nov.
  2021.

\bibitem{Mohammadi2023Network}
M.~Mohammadi, T.~T. Vu, H.~Q. Ngo, and M.~Matthaiou, ``{Network-Assisted
  Full-Duplex Cell-Free Massive MIMO: Spectral and Energy Efficiencies},''
  \emph{IEEE J. Sel. Areas Commun.}, vol.~41, no.~9, pp. 2833--2851, Sept.
  2023.

\bibitem{Sun2023Hierarchical}
X.~Sun, C.~Sun, J.~Li, D.~Wang, H.~Zhang, Y.~Hao, and X.~You, ``{Hierarchical
  Reinforcement Learning for AP Duplex Mode Optimization in Network-Assisted
  Full-Duplex Cell-Free Networks},'' \emph{IEEE Syst. J.}, vol.~17, no.~4, pp.
  6540--6551, Dec. 2023.

\bibitem{liu2024cooperative}
\BIBentryALTinterwordspacing
S.~Liu, M.~Li, and Q.~Liu, ``{Joint BS mode Selection and Beamforming Design
  for Cooperative Cell-Free ISAC Networks},'' \emph{arXiv:2305.10800}, 2023.
  [Online]. Available: \url{http://arxiv.org/pdf/2305.10800}
\BIBentrySTDinterwordspacing

\bibitem{elfiatoure2023cellfree}
M.~Elfiatoure, M.~Mohammadi, H.~Q. Ngo, and M.~Matthaiou, ``{Cell-Free Massive
  MIMO for ISAC: Access Point Operation Mode Selection and Power Control},'' in
  \emph{2023 IEEE Globecom Workshops (GC Wkshps)}, Dec. 2023, pp. 104--109.

\bibitem{Mao2023An}
Y.~Mao, Y.~Huang, X.~Yu, Y.~Xin, Y.~Wang, and W.~Hong, ``{An Radio Frequency
  Interference Mitigation Approach for Spaceborne SAR System in Low SINR
  Condition},'' \emph{IEEE Trans. Geosci. Remote Sens.}, vol.~61, pp. 1--14,
  Oct. 2023.

\bibitem{Liu2023Cluster}
Y.~Liu, J.~Zhang, Y.~Zhang, Z.~Yuan, and G.~Liu, ``{A Shared Cluster-based
  Stochastic Channel Model for Integrated Sensing and Communication Systems},''
  \emph{IEEE Trans. Veh. Technol.}, pp. 1--13, May 2024.

\bibitem{Sakhnini2021Bound}
A.~Sakhnini, M.~Guenach, A.~Bourdoux, and S.~Pollin, ``{A Cram{\'e}r-Rao Lower
  Bound for Analyzing the Localization Performance of a Multistatic Joint
  Radar-Communication System},'' in \emph{IEEE Int. Online Symp. Jt. Commun.
  Sens. (JC\&S)}, Feb. 2021, pp. 1--5.

\bibitem{Guo2021Performance}
Y.~Guo, C.~Li, C.~Zhang, Y.~Yao, and B.~Xia, ``{Performance analysis of the
  full-duplex joint radar and communication system},'' in \emph{IEEE/CIC Int.
  Conf. Commun. China. (ICCC)}, Jul. 2021, pp. 505--510.

\bibitem{Arulkumaran2017Deep}
K.~Arulkumaran, M.~P. Deisenroth, M.~Brundage, and A.~A. Bharath, ``{Deep
  Reinforcement Learning: A Brief Survey},'' \emph{IEEE Signal Process Mag.},
  vol.~34, no.~6, pp. 26--38, Nov. 2017.

\bibitem{Lv2018Uplink}
Q.~Lv, J.~Li, P.~Zhu, and X.~You, ``{Spectral Efficiency Analysis for
  Bidirectional Dynamic Network With Massive MIMO Under Imperfect CSI},''
  \emph{IEEE Access}, vol.~6, pp. 43\,660--43\,671, 2018.

\bibitem{Van2020Massive}
T.~Van~Chien, E.~Björnson, and E.~G. Larsson, ``{Joint Power Allocation and
  Load Balancing Optimization for Energy-Efficient Cell-Free Massive MIMO
  Networks},'' \emph{IEEE Trans. Wireless Commun.}, vol.~19, no.~10, pp.
  6798--6812, Oct. 2020.

\bibitem{Zhang2022Time}
Q.~Zhang, H.~Sun, X.~Gao, X.~Wang, and Z.~Feng, ``{Time-Division ISAC Enabled
  Connected Automated Vehicles Cooperation Algorithm Design and Performance
  Evaluation},'' \emph{IEEE J. Sel. Areas Commun.}, vol.~40, no.~7, pp.
  2206--2218, Jul. 2022.

\bibitem{Ma2023Wideband}
L.~Ma, J.~Lai, Y.~Yin, C.~Xia, C.~Gu, and J.~Mao, ``{A Wideband Co-Linearly
  Polarized Full-Duplex Antenna-in-Package With High Isolation for Integrated
  Sensing and Communication},'' \emph{IEEE Antennas Wirel. Propag. Lett.},
  vol.~22, no.~9, pp. 2185--2189, Sept. 2023.

\bibitem{Zhu2022Load}
Y.~Zhu, J.~Li, P.~Zhu, D.~Wang, H.~Ye, and X.~You, ``{Load-Aware Dynamic Mode
  Selection for Network-Assisted Full-Duplex Cell-Free Large-Scale Distributed
  MIMO Systems},'' \emph{IEEE Access}, vol.~10, pp. 22\,301--22\,310, 2022.

\bibitem{Heath2011Multiuser}
R.~W. Heath~Jr, T.~Wu, Y.~H. Kwon, and A.~C.~K. Soong, ``{Multiuser MIMO in
  Distributed Antenna Systems With Out-of-Cell Interference},'' \emph{IEEE
  Trans. Signal Process.}, vol.~59, no.~10, pp. 4885--4899, Oct. 2011.

\bibitem{Zhu2011Performance}
H.~Zhu, ``{Performance Comparison Between Distributed Antenna and Microcellular
  Systems},'' \emph{IEEE J. Sel. Areas Commun.}, vol.~29, no.~6, pp.
  1151--1163, Jun. 2011.

\end{thebibliography}

\vfill

\end{document}